%% file: stream.tex
\newif\ifFULL
\FULLtrue % for full version
\ifFULL
\documentclass[10pt,a4]{article}
\newif\ifCLASSINFOpdf\CLASSINFOpdffalse
\newcommand\IEEEauthorblockN[1]{#1\\[-0.5em]}
\newcommand\IEEEauthorblockA[1]{\parbox{4cm}{\center\normalsize #1}}

\newcommand\qed{\hfill$\Box$}
\newenvironment{proof}{{\bf Proof.}}{\qed\\

}
\else
\documentclass[conference]{IEEEtran}
\IEEEoverridecommandlockouts
\fi % FULL
\usepackage{syntax, semantics, sublabels}
\newcommand\ForClause[3]{\ensuremath{\syntax{\term{for}\ \term{\$}#1\ \term{in}\ #2\ \term{return}\ #3}}}
\newcommand\LetClause[3]{\ensuremath{\syntax{\term{let}\ \term{\$}#1\ \term{:=}\ #2\ 
\term{return}\ #3}}}

\DeclareMathAlphabet{\mathtt}{OT1}{cmtt}{m}{n}
\SetMathAlphabet{\mathtt}{bold}{OT1}{cmtt}{m}{n}

\usepackage{xspace}
\newcommand{\minXQuery}{MinXQuery\xspace}

\usepackage{amsmath,amstext,amssymb}
\usepackage{graphicx}
\usepackage{color}
\usepackage{url}

\newtheorem{lemma}{Lemma}
\newtheorem{theorem}{Theorem}

\newtheorem{definition}{Definition}
\usepackage{fancybox}
\ifCLASSINFOpdf
\else
\fi
\usepackage[tight,footnotesize]{subfigure}
\input macros.tex
\def\mft{\text{{\scshape mft}}}
\def\qed{\hfill$\Box$}
\def\ft{\text{{\scshape ft}}}
\def\mtt{\text{{\scshape mtt}}}
\def\tt{\text{{\scshape tt}}}
\def\eval{{\sf eval}}
\def\EVAL{\text{{\scshape eval}}}
\newcommand\m[1]{\mathit{#1}}
\newcommand{\sem}[1]{[\![ #1 ]\!]}
\newcommand{\angl}[1]{\langle #1\rangle}

\usepackage{bbold}
\newcommand\comp{\mathbin{\text{\raisebox{-0.5pt}{\parbox{3pt}{\textbb{.}\\[-7.5pt]\textbb{,}}}}}}
\begin{document}
\title{XQuery Streaming by Forest Transducers%
\ifFULL\thanks{This is the full version of the paper published in the proceedings of ICDE 2014.}
\fi %FULL
}


\author{
\IEEEauthorblockN{Shizuya Hakuta}
\IEEEauthorblockA{DOCOMO Datacom, Inc.\\Japan}
\and
\IEEEauthorblockN{\ifFULL\else\ \ \fi Sebastian Maneth\ifFULL\else${}^*$\fi%
\thanks{\ifFULL\else${}^*$\fi
The present affiliation of this author is
        University of Edinburgh (UK).}
}
\IEEEauthorblockA{University of Oxford\\UK}
%School of Electrical and\\Computer Engineering\\
%Georgia Institute of Technology\\
%Atlanta, Georgia 30332--0250\\
%Email: http://www.michaelshell.org/contact.html
\and
\IEEEauthorblockN{Keisuke Nakano and Hideya Iwasaki}
\IEEEauthorblockA{
The University of Electro-Communications\\ Japan
%Twentieth Century Fox\\
%\and
%\IEEEauthorblockN{James Kirk\\ and Montgomery Scott}
%\IEEEauthorblockA{Starfleet Academy\\
%San Francisco, California 96678-2391\\
}}


%Telephone: (800) 555--1212, Fax: (888) 555--1212}

\maketitle
%\footnote{The present affiliation of this author is
%          University of Edinburgh (UK).}

\begin{abstract}
Streaming of XML transformations is a challenging task and only
a few existing systems support streaming. Research approaches generally define
custom fragments of XQuery and XPath that are amenable to streaming, and then
design custom algorithms for each fragment. These languages have several
shortcomings. Here we take a more principled approach to the problem of
streaming XQuery-based transformations. We start with an elegant transducer
model for which many static analysis problems are well-understood: the
Macro Forest Transducer (MFT). We show that a large fragment of XQuery can
be translated into MFTs --- indeed, a fragment of XQuery, that can express important
features that are missing from other XQuery stream engines, such as GCX:
our fragment of XQuery supports XPath predicates and let-statements. 
We then use an existing streaming engine for MFTs and apply
a well-founded set of optimizations from functional programming such as
strictness analysis and deforestation.  
Our prototype achieves time and memory efficiency comparable to the fastest 
known engine for XQuery streaming, GCX.
This is surprising because our engine relies on the OCaml built in garbage collector
and does not use any specialized buffer management,
while GCX's efficiency is due to clever and explicit buffer management.
\end{abstract}


%\begin{IEEEkeywords}
%XQuery Stream, finite state tree transducer, composition, query optimization.
%\end{IEEEkeywords}

% \fi
%\IEEEpeerreviewmaketitle
\ifFULL\maketitle\fi

\section{Introduction}
Data is often transmitted as a continuous stream; e.g.,
sensor readings such as weather data, or text messages such as news feeds.
Streams are also used to transmit large data that does not fit into memory.
Imagine now to query or \emph{transform} streamed data, and,  
that the data is tree structured (e.g. in XML or JSON). 
Doing this within limited memory is a challenging task.
%True constant-memory (only depending on the query) algorithms
%are rare and only exist for very limited query/transformation formalism;
%essentially, they correspond to running a finite state word transducer or
%automaton, as considered by Segoufin et al~\cite{DBLP:conf/pods/SegoufinV02,DBLP:conf/icdt/SegoufinS07}.
Only a few systems support streaming of XML transformations.
These systems work on a ``best effort'' basis and try to use
as little memory as possible, but do not give guarantees on the amount
of memory used. XML transformations are conveniently expressed in XQuery and
XSLT. An example of a best-effort streaming engine for XSLT is 
Kay's Saxon~\cite{DBLP:journals/debu/Kay08}.

Research approaches generally proceed by defining custom fragments of XQuery and
XPath that are amenable to streaming, and then design custom algorithms
for each fragment. 
Examples are forward XPath of Olteanu's Spex~\cite{DBLP:journals/tkde/Olteanu07} and
the XQuery fragment of Koch, Scherzinger, and 
Schmidt's GCX~\cite{DBLP:conf/vldb/KochSS07}. GCX is the fastest XQuery streaming engine 
that we know.
These languages have several shortcomings:
they lack expressiveness needed for important transformations (for instance, they cannot express
XPath predicates). Since these languages have been designed specifically for streamability they
are difficult to analyze and less amenable to more general optimizations, other than those
optimizations specifically engineered for the streaming problems. For instance, their properties
with respect to sequential composition are not well-understood, hence they cannot easily be used in conjunction with
other transformations and filters outside of the fragment. 

In this paper, we take a different and more principled approach to the problem of
streaming XQuery-based transformations. We start with an elegant transducer model
that has been studied in the literature extensively, and for which many static analysis
problems are well-understood, namely the \emph{Macro Forest Transducer} 
of Perst and Seidl~\cite{DBLP:journals/ipl/PerstS04}. 
Known properties for this model include:
effective exact type checking, 
composition closure with many known classes of transformations,
decidable strictness analysis, 
decidable equivalence check for transformations of linear size increase, 
and the possibility to represent outputs succinctly using grammar-based compression.
We show that a large fragment of XQuery can be translated into MFTs --
indeed a fragment of XQuery, that can express the important features missing from the GCX
language, such as XPath predicates and let-expressions. We rely on a streaming 
execution engine for MFTs implemented in OCaml by Nakano and Mu~\cite{DBLP:conf/aplas/NakanoM06}.
We use a well-founded set of MFT-optimizations that are common in functional programming.
In particular, we apply deforestation~\cite{DBLP:journals/tcs/Wadler90}
which has been heavily studied and applied in functional programing languages like Haskell.
Thus, we have moved XQuery streaming from a ``one off'' problem to
a well-understood problem within transducer programming. 
Further, we can exploit feature of MFTs not present in our XQuery fragment.
For instance, MFTs naturally support recursive function definitions.
It is thus possible to translate a given XQuery program into an MFT, and to then
change the MFT by adding recursive definitions. Alternatively, it would be possible
to write a recursive MFT program which uses small XQuery programs in the right-hand
sides of its rules. Another convenient feature of MFTs is their ability to validate
the input, during transformation. This allows to check a XML Schema or Relax NG
in one pass during the streaming transformation.

Our contributions are summarized as:
\begin{itemize}
\item We formalize a translation from a fragment of XQuery into
macro forest transducers. We implemented a prototype system of
our translation.
\item We demonstrate the efficiency of our system and compare it experimentally
with GCX. In a nutshell, our system performs on par with GCX.
\item We present three different optimizations: two forms of parameter removal,
and the removal of stay moves. Together, they often induce a speedup of one
order of magnitude.
\item We prove bigO complexity of composition constructions for
MFTs. In the literature these constructions have exponential time complexity.
Using a particular feature of our MFTs (``stay moves'') we are able
to prove a quadratic time complexity. Roughly speaking, 
stay moves allow to compress intermediate rule trees.
\end{itemize}

Let us now explain in more detail some of the useful known properties about macro
forest transducers that we use. 
They can be applied to XQuery programs
thanks to our translation of XQuery into MFT.

(1) Streaming:
MFTs traverse an input forest by applying rules 
based on structural recursion.
Nakano and Mu~\cite{DBLP:conf/aplas/NakanoM06} show that
any MFT-style program can be streamed thanks to the 
structural recursion restriction.
Their streaming is based on
the composition of an MFT and an XML parsing transducer
in a way similar to that of a macro tree transducer and a top-down tree transducer.
They illustrate that the obtained transducer can be naturally implemented
as a pushdown machine which directly processes an input XML stream
and produces the output XML stream.
Their streaming approach has, however, an disadvantage that
it is hard for programmers to write an MFT-style program for XML transformation.
Our XQuery-to-MFT translation makes up for this shortcoming.
Programmers can now write transformations in the user-friendly query language 
XQuery instead of MFT-style programs, to obtain an XML stream processors.
Due to the expressiveness of MFTs, our XQuery streaming supports
even complex queries which contain nested loops and multiple variable accesses.

(2) Composition:
MFTs themselves are \emph{not} closed under composition. Several important subclasses
however are closed under composition. For instance, their restriction to linear size
increase are MSO definable and therefore are closed under composition; this follows from
Engelfriet and Maneth's result that macro tree transducers of linear size increase
are MSO definable~\cite{DBLP:journals/siamcomp/EngelfrietM03} and Maneth's result that 
the composition hierarchy of macro tree transducers
collapses for linear size increase~\cite{DBLP:conf/fsttcs/Maneth03}.
Note also that the linear size property is decidable for MFTs.
If the MFT does not use context-parameters then it is called a
``top-down forest transducer'' (FTs).
FTs are also not closed under composition, however, we show that
the composition of two FTs can be realized by one MFT.
We give an explicit construction for this result, and show that its time complexity is $O(|\Sigma||M_1||M_2|)$,
where $\Sigma$ is the size of the alphabet of element/attribute and text constants used by the
transducers.
This is possible in the presence of stay-moves (or by using DAG compressed right-hand sides); 
the classical constructions for composition of top-down tree transducers 
by Rounds~\cite{DBLP:journals/mst/Rounds70} 
and Baker~\cite{DBLP:journals/iandc/Baker79b} are in fact exponential in the size of 
the first transducer $M_1$. 

(3) Static Analyses:
A powerful feature of MFTs is inverse type inference: regular tree languages
are effectively preserved by inverses of MFT translations. 
This allows to perform exact type checking~\cite{DBLP:conf/pods/ManethBPS05,DBLP:conf/icdt/ManethPS07}.
It also allows to check the parameter strictness of the states of a transducer.
We use a simple version of strictness analysis in this paper in order to reduce
the number of context parameters of an MFT; this is simular to deaccumulation, see~\cite{DBLP:journals/jlp/GieslKV07}.
Parameter reduction allows to greatly improve
the efficiency of the MFTs that are obtained via our translation from XQuery.
Another important static analysis (not used here) is equivalence checking: if two MFTs
are of linear size increase then their equivalence 
is decidable~\cite{DBLP:journals/ipl/EngelfrietM06}.

%\smallskip
%\subsection*{Related Work}
{\bf Related Work.}\quad
Streaming of XPath has been studied extensively.
The fundamental article of Green et al.~\cite{DBLP:journals/tods/GreenGMOS04} shows how
to translate filter-less descendant/child-XPath queries into finite state word automata (DFAs).
The DFAs are executed top-down through the tree, using memory proportional to
the depth of the XML document tree.
Several systems are based on finite automata, such as Xfilter, XTrie, YFilter, PrefixFilter,
AFilter, and the XPush machine. Streaming of XPath queries that include filters is 
difficult, because candidate nodes that depend on a filter need to be stored in memory.
Recent work of Niehren et al. studies this in detail,
see, e.g.,~\cite{DBLP:journals/iandc/GauwinNT11,DBLP:conf/wia/GauwinN11,DBLP:conf/wia/DebarbieuxGNSZ13}.
Bar-Yossef, Fontoura, and Josifovski~\cite{DBLP:journals/jcss/Bar-YossefFJ07}
study theoretical bounds as well as practical streaming algorithms for XPath.
Shalem and Bar-Yossef~\cite{DBLP:conf/icde/ShalemB08} study twig-join algorithms
over XML streams. 

Early systems for XQuery streaming include the BEA processor~\cite{DBLP:journals/vldb/FlorescuHKLRWCS04},
FluxQuery~\cite{DBLP:conf/vldb/KochSSS04a}, and the Raindrop system~\cite{DBLP:journals/dke/WeiRML08}.
Streaming of XQuery based on physical algebra operators is presented in~\cite{DBLP:conf/icde/FernandezMSS07}.
There has been a transducer-based approach for streaming of XQuery~\cite{DBLP:conf/vldb/LudascherMP02},
but, the transducers are more restrictive than our macro forest transducer.
Streaming of XSLT has been considered by Dvorakova~\cite{DBLP:journals/informaticaSI/Dvorakova08} 
using tree transducers.
The transducers are similar to tree-walking tree transducers (see e.g.,~\cite{tw})
as they can use XPath expressions with forward and backward axes 
in the right-hand sides of their rules. For a given 
XSLT program, they present an analysis that attempts to find a smallest class of transducers
that captures the given transformation. The classes are distinguished by the number of passes and
the memory needed. No experimental evaluation of Dvorakova's work is available.
A mature commercial XSLT engine is SAXON by Kay~\cite{DBLP:journals/debu/Kay08}. 
Its performance is lower with respect to GCX and our engine, 
but, SAXON is not comparable to these systems because it implements the full
W3C standard of XSLT while GCX and our engine are proof-of-concept prototypes 
supporting restricted subsets of XQuery.
Note that the main purpose of the new XSLT 3.0 specification is to support 
streaming.
For this, new primitives and modes are introduced
for indicating the desire to stream.
A set of rules determines if the given program can indeed be streamed.
Their memory requirement is that ``not all input and output nodes are held in memory''. 

There are several works on streaming XML transformations
based on programming language theoretic approaches.
In a direction similar to streaming MFTs~\cite{DBLP:conf/aplas/NakanoM06},
Frisch and Nakano proposed stream processing for term rewriting systems (TRS)~\cite{DBLP:conf/planX/FrischN07},
which is more powerful than MFTs because of their Turing completeness.
However, it is still hard to write an XML transformation in TRS
and it requires careful programming for an efficient stream processing.
Kodama, Suenaga, and Kobayashi applied type theories to obtain stream processing~\cite{DBLP:journals/jfp/KodamaSK08}.
They employ an ordered linear type system for guaranteeing the possibility of streaming.
In both approaches,
the transformation must be written in their own programming languages 
instead of existing XML processing languages like XQuery and XSLT.

\section{Preliminaries}

An XML document represents an ordered unranked tree.
The nodes of this tree are of different types:
element nodes, text nodes, attribute nodes, processing instructions, etc.
Here we distinguish three types: element nodes, attribute nodes,
and text nodes. Our techniques easily extend to other types of nodes.
Element nodes have an arbitrary number of children and
are written in XML as
\begin{quote}\hspace*{-1em}
\|<|{\it elementname} $a_1$\|=|$v_1$ \dots ~$a_n$\|=|$v_n$\|>|\dots\|</|{\it elementname}\|>|
\end{quote}
%\[
%\text{\tt <}\text{\it elementname} ~a_1=v_1 .. a_n=v_n>...<\!/\text{\it elementname}>
%\]
where {\it elementname\/} is the name of the element node, 
$a_i$ are attribute names, $v_i$ are text values (given as strings between double quotes), 
and ``...'' (possibly) contains
further descendant nodes of this element node.
In our unranked tree model,
the first $n$ children of an element node 
are attribute nodes labeled $a_1,\dots,a_n$; 
each attribute node has exactly one child which is a text node labeled $v_i$.
Text nodes have no children, i.e., they are leaves of the tree;
they are labeled by a ``text content'' which is a
character sequence appearing in the XML document.
For instance, this XML snippet
\begin{verbatim}
  <book isbn="123" price="$99"><author>Knuth
  </author><title>Art of Programming</title>
  </book>
\end{verbatim}
represents the unranked tree in Figure~\ref{fig_1}.
It consists of a root node labeled ``book'' which has
four children,
labeled ``isbn'', ``price'', ``author'', and ``title'', respectively.
Each of these children nodes has exactly one child that is a
text node (labeled ``123'', ``\$99'', etc).
\begin{figure}[ht]
{\centerline{
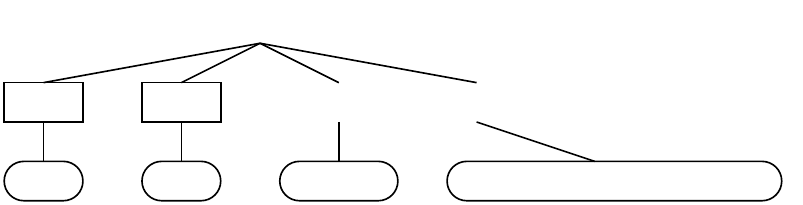
}}
\caption{An XML forest; text nodes are circled and attribute nodes
are boxed.}\label{fig_1}
\end{figure}
Formally, each node has a type \emph{and} a name. 
For us, both are part of the label, i.e., labels
are of the form (type,name).
We abstract from such pairs and assume that every node is
labeled by a word in $\mathbf{U}^+$. The set {\bf U}
is a fixed universal (finite) alphabet {\bf U} of characters.
For simplicity,
we represent text nodes as special element nodes with empty child.
We define \emph{XML forests} as sequences of unranked trees.

\begin{definition}\label{def:forest}
An \emph{XML forest} is a sequence $t_1 \cdots t_n$ where
$n\geq 0$ and $t_1,\dots,t_n$ are unranked trees.
An \emph{unranked tree} consists of a root node labeled by a word in $\mathbf{U}^+$
and a (possibly empty) sequence of subtrees.
The set of all XML forests is denoted by $\mathcal{F}$.
We often write a forest in term notation, i.e., 
generated by this EBNF:
\begin{bnf}
\abovedisplayskip=2pt
\belowdisplayskip=2pt
\begin{align*}
\nonterm{forest} \=&\mathrel{}\varepsilon \mid \nonterm{tree}\,\nonterm{forest}\\[-2pt]
\nonterm{tree} \=&\mathrel{}\nonterm{label}(\nonterm{forest})\\[-2pt]
\nonterm{label} \=&\mathrel{}\nonterm{\bf U}^+.
\end{align*}
\end{bnf}
\end{definition}

\subsection{\minXQuery: XQuery Fragment}

We consider a downward navigational fragment of XQuery, called \minXQuery.
XPath expressions in our queries may use the child, descendant, and following-sibling
axes. Filters may test the existence of a path, or may compare 
a text node or attribute value against a constant string value.
We do not allow where-clauses, ``ordered by''-statements, 
recursive function definitions, and joins.
Thus, \minXQuery expressions consist of nested for-loops and let-statements.
We do not discuss text predicates such as 
starts-with and contains; they are easy to support and 
are part of our prototype.

Figure~\ref{fig:minxquery-syntax} shows an EBNF
of \minXQuery fragment.
\begin{figure}[bt]
\begin{small}
  \centering
  \input{figure-minxquery-syntax}\vspace*{-\baselineskip}
  \caption{Syntax of \minXQuery}\label{fig:minxquery-syntax}
\end{small}
\end{figure}
Additional to the syntax defined in that figure we
impose the following restrictions on an XQuery program:
\begin{itemize}
 \item The input document is bound to a special variable
       named $\syntax{\term{\$input}}$.
 \item Every XPath expression starts with a variable %.
       which has been introduced in the nearest enclosing
       \texttt{for} clause, or, 
       if no such \syntax{\term{for}} clause exists, with 
       the variable $\syntax{\term{\$input}}$.
\end{itemize}
\noindent
Although the second restriction disables to define join queries,
we can still write many practical queries 
by utilizing XPath predicates and nested loops.
We do not define the semantics of XQuery programs here;
see, e.g.,~\cite{DBLP:conf/afp/Wadler02}.
As an example consider the following \minXQuery program:
\begin{verbatim}
  for $v1 in $input/descendant::a return
    for $v2 in $v1/descendant::b return
      let $v3 := $v2/descendant::c return
      let $v4 := $v2/descendant::d return
      ($v1,$v2,$v3,$v4)
\end{verbatim}
Consider the evaluation of this program on the
following input document:
\begin{verbatim}
  <doc><a><b><c><c/></c><d/><d/></b>
  <b><d/></b></a></doc>
\end{verbatim}
Let us refer to the first $b$-node in the document by $b_1$,
and to the second one (in preorder) by $b_2$, etc.
The first sequence of subtrees that are output
are those rooted at $a_1$, $b_1$, $c_1$, $c_2$, $d_1$, and $d_2$, respectively.
Another sequence of subtrees is also output,
whose roots are $a_1$, $b_2$, and $d_3$.
%Three more quadruple of subtrees are output, namely at 
%$(a_1,b_1,c_1,d_2)$,
%$(a_1,b_1,c_2,d_1)$,
%and $(a_1,b_1,c_2,d_2)$.

\subsection{Macro Forest Transducers}

A forest transducer is a finite state machine that takes as
input an XML forest and produces as output an XML forest.
Recall that each node of an XML forest is labeled by a non-empty
word over {\bf U}.
In our transducers we abstract from {\bf U}-characters
forming the label of a node by fixing a finite set $\Sigma$ of words in $\text{\bf U}^+$ that
are of interest to the transducer.
We refer to elements of $\Sigma$ as ``symbols''.
A rule for the state of a transducer tests if the current input node is 
labeled by a given symbol $\sigma\in\Sigma$.
A state also has a ``default rule'' which is applicable if no other rule
of that state applies; the default rule applies to any 
${\mathbf U}^+$-labeled node. 
The parsing of an input forest according to the EBNF given in
Definition~\ref{def:forest} also provides us with the information whether
the empty forest $\varepsilon$ is reached.
For instance, the parsing of the forest $a(b())$ is parsed as 
$a(b(\varepsilon)\varepsilon)\varepsilon$. Our transducers may use this
information on occurrences of $\varepsilon$: for each state we require
that the transducer has a rule for the input $\varepsilon$.

We define the fixed set $Y=\{y_1,y_2,\dots\}$ of \emph{context parameters},
also called \emph{accumulating parameters}, or simply \emph{parameters}.
A \emph{ranked set} is a set together with a mapping that associates
to each element of the set a non-negative number, called the rank of that
element. For a ranked
set $Q$ we denote by $Q^{(k)}$ the subset of symbols which have rank~$k$. 

\begin{definition}\normalfont
Let $\Sigma\subseteq\text{\bf U}^+$ be a finite set of symbols
and let $\%t$ be a special symbol not in $\Sigma$.
A \emph{macro forest transducer} $M$ over $\Sigma$ 
is a tuple $(Q,\Sigma,q_0,R)$ where 
$Q$ is a finite ranked set of \emph{states}, each 
of rank $\geq 1$.
The \emph{initial state} $q_0$ is in $Q^{(1)}$.
Let $q\in Q^{(m+1)}$ with $m\geq 0$.
For every input symbol $\sigma\in\Sigma$
the set $R$ contains at most one \emph{(q,$\sigma$)-rule} 
of the form 
%\begin{normalfont}
%\begin{align*}
\[
q(\sigma(x_1)x_2, y_1, \dots, y_m)\to r\text{,}
\]
%\end{normalfont}
where $r$ is a forest over $\Sigma\cup Q\cup X\cup Y_m$ with
$X=\{x_0,x_1,x_2\}$,
$Y_m=\{y_1,\dots,y_m\}$, and the properties that
a leaf has a label in $X$  
if and only if it is the first child of a $Q$-node, and parameters
in $Y_m$ may only appear at leaves.
Variables $x_0$, $x_1$ and $x_2$ bind 
the current node with the rest of the stream,
the children of the current node, and
the rest of the stream, respectively.
Additionally, $R$ contains exactly one rule of each
of the following two kinds:
(1)
a \emph{$\varepsilon$-rule} of the form
\[
q(\varepsilon, y_1, \dots, y_m)\to r\text{,}
\]
where $r$ is as for $(q,\sigma)$-rules, but with 
$X=\{x_0\}$;
(2) a \emph{default rule} of the form
\[
q(\%t(x_1)x_2, y_1, \dots, y_m)\to r\text{,}
\]
where $r$ is as for $(q,\sigma)$-rules but where
binary nodes may be labeled $\%t$.
\qed
\end{definition}

Note that our transducers are deterministic by definition,
i.e., for a state at a given input node, at most one rule is applicable.
Note also that they are total and define an output for any arbitrary
given input forest, due to the presence of the default rules.

Let $M=(Q,\Sigma,q_0,R)$ be an MFT. 
A call of the from $q'(x_0)$ in the right-hand side of
a rule of $M$ is called ``stay move''. Note that stay
moves can  give rise to non-terminating computation.
For instance, the rule
$q(\varepsilon)\to q(x_0)$ if called on the leaf
$\varepsilon$ does not terminate.
We do not further formalize termination but refer the
reader to Section~5.2 of~\cite{DBLP:journals/acta/EngelfrietM03} 
where this is discussed for a formalism similar to MFTs,
namely for deterministic pebble macro tree transducers.
We only deal with terminating MFTs; all our constructions
operate on terminating MFTs, and are guaranteed to construct
terminating MFTs.
Therefore, we always mean ``terminating MFT'' from now on,
when we speak about MFT.

We define the semantics of the MFT $M$.
For a given input XML forest $f$, $M$'s output 
denoted $\sem{M}(f)$ is defined as $\sem{q_0}(f)$. 
Every state $q\in Q^{(m+1)}$, $m\geq 0$ realizes the function 
$\sem{q}:\mathcal{F}^{m+1}\to\mathcal{F}$ defined recursively 
for forests $g_0,f_1,\dots,f_m\in\mathcal{F}$ as
\[
\sem{q}(g_0,f_1,\dots,f_m)=\sem{r}
\]
where $r$ is the right-hand side of the unique rule 
that is applicable:
(i) if $g_0=\varepsilon$ then $r$ is the right-hand side 
of $q$'s $\varepsilon$-rule.
(ii) If $g_0=\sigma(g_1)g_2$ for forests $g_1,g_2$, then 
$r$ is the right-hand side of the $(q,\sigma)$-rule of $M$,
if it exists, and otherwise is the right-hand side of the
default rule of $q$.
The tree $\sem{r}$ is defined inductively as:
$\sem{\varepsilon}=\varepsilon$, $\sem{y_j}=f_j$ if $j\in\{1,\dots,m\}$, and 
$\sem{q'(x_i,u_1,\dots,u_n)}=\sem{q'}(g_i,\sem{u_1},\dots,\sem{u_n})$
if $q'\in Q^{n+1}$, $i\in\{0,1,2\}$, 
and $u_1,\dots,u_n$ are subtrees in $r$.

Let us consider this example query $P_{\text{person}}$:
\begin{verbatim}
 <out>{ for $b in
        $input/person[./p_id/text() = "person0"] 
        return let $r := $b/name/text()
        return $r }</out>
\end{verbatim}
The query selects all text-node children of any name-node, 
that is child of a person node, and that person-node has a 
p\_id-child with text node child of content ``person0''.
Here are the rules of the MFT $M_{\text{person}}$ of this example; it outputs
a root-node ``out'' where the children are the results of the above
query.
%\begin{normalfont}
\[
\begin{array}{lcl}
q_0(\%(x_1)x_2)&\to&\text{out}(q_1(x_0))\\
%q_0(\varepsilon)&\to&\varepsilon\\

q_1(\text{person}(x_1)x_2)&\to& q_2(x_1,q_4(x_1))\, q_1(x_2)\\
q_1(\%t(x_1)x_2)&\to& q_1(x_1)\, q_1(x_2)\\
%q_1(\varepsilon)&\to&\varepsilon\\

q_2(\text{p\_id}(x_1)x_2,y_1)&\to& q_3(x_1,y_1,q_2(x_2,y_1))\\
q_2(\%t(x_1)x_2,y_1)&\to& q_2(x_2,y_1)\\
%q_2(\varepsilon)&\to&\varepsilon\\

q_3(\text{person0}(x_1)x_2,y_1,y_2)&\to& y_1\\
q_3(\%t(x_1)x_2,y_1,y_2)&\to& q_3(x_2,y_1,y_2)\\
q_3(\varepsilon,y_1,y_2)&\to& y_2\\

q_4(\text{name}(x_1)x_2)&\to& q_5(x_1)\, q_4(x_2)\\
q_4(\%t(x_1)x_2)&\to& q_4(x_2)\\
%q_4(\varepsilon)&\to&\varepsilon\\

q_5(\%t_{\rm text}(x_1)x_2)&\to& \%t(\varepsilon)\,q_5(x_2)\\
q_5(\%t(x_1)x_2)&\to& q_5(x_2)\\
q_i(\varepsilon)&\to&\varepsilon\quad\text{for }i\in\{1,2,4,5\}
\end{array}
\]
%\end{normalfont}\\[-2.1ex]
where the pattern $\%t_{\rm text}(x_1)x_2$ matches only text nodes,
thereby $x_1$ should bind $\varepsilon$ for `normal' inputs.
%\noindent
Let us run the MFT $M_{\text{person}}$ on the XML tree for this document:
\begin{verbatim}
  <person><p_id><a/>person0</p_id><name>Jim
  </name><c/><name>Li</name></person>
\end{verbatim}
We want to compute 
$\sem{q_0}(\text{person}(t_1\,t_2\,t_3\,t_4\, \varepsilon)\,\varepsilon)$ 
where $t_1$ is a tree with root-node labeled p\_id,
$t_2,t_4$ are name-nodes, and $t_3$ is a $c$-leaf.
According to the first rule we obtain 
$\text{out}(\sem{q_1}(\text{person}(f)\,\varepsilon))$,
where $f$ is the forest $t_1\, t_2\, t_3\, t_4\, \varepsilon$.
We apply the $(q_1,\text{person})$-rule to obtain
\[
\sem{q_2}(f,\sem{q_4}(f))\,\sem{q_1}(\varepsilon)\text.
\]
We remove the $q_1$-call because it produces $\varepsilon$.
Let $t_1=\text{p\_id}(f_1)$ and $f_2=t_2\,t_3\,t_4\,\varepsilon$. 
According to the $(q_2,\text{p\_id})$-rule, we obtain
$\sem{q_3}(f_1,\sem{q_4}(f),\sem{q_2}(f_2,\sem{q_4}(f)))$
Let $f_1=\text{a}(\varepsilon)\,\text{person0}(\varepsilon)\,\varepsilon$,
we apply the default rule for $q_3$ to obtain
\[
\sem{q_3}(\text{person0}(\varepsilon)\,\varepsilon,\sem{q_4}(f),\sem{q_2}(f_2,\sem{q_4}(f)))\text.
\]
Then we apply 
the $(q_3,\text{person0})$-rule to obtain $\sem{q_4}(f)$.
Let $t_2=\text{name}(f_3)$ and $t_4=\text{name}(f_4)$.
Since $f=\text{p\_id}(f_1)\,\text{name}(f_3)\,\text{c}(\varepsilon)\,\text{name}(f_4)\,\varepsilon$,
we apply
the $q_4$-rules to obtain
$\sem{q_5}(f_3)\,\sem{q_5}(f_4)\,\varepsilon$.
The forest $f_3$ only consists of the ``Jim'' text node,
which is thus output by $q_5$.
Similarly, the $q_5$-call is applied to the other text in $f_4$
and output it.
Thereby our final output is $\text{out}(\text{Jim}\,\text{Li})$
where the empty forest $\varepsilon$ is omitted.
Note that in XML documents it is not possible to have two text
nodes that are direct siblings of each other.
Thus, our MFT processor for this example outputs
\verb!<out>JimLi</out>!.

It is interesting to observe why the state $q_3$ of the transducer $M_{\text{person}}$
uses two parameters $y_1$ and $y_2$.
This is done to simulate the
existential semantics of XPath filters (the two parameters are used as two branches
of a if-then-else statement).
To see this, consider the input
\begin{verbatim}
  <person><p_id><a/>perso7</p_id><name>Jim
  </name><c/><p_id>person0</p_id></person>
\end{verbatim}
After some initial steps (as before) we obtain
\[
\sem{q_3}(\text{perso7}(\varepsilon)\,\varepsilon,\sem{q_4}(f'),\sem{q_2}(f_2',\sem{q_4}(f')))\text.
\] 
%where $f_1'$ is the forest $\text{perso7}(\varepsilon)\,\varepsilon$.
Since the filter is not true at this point (because ``perso7'' is
not equal to ``person0''), we need to check if further 
\text{p\_id}-children of the person-node satisfy the filter. 
We carefully prepared for this event, by supplying $q_3$'s initial
call with a second parameter that contains 
a $q_2$-call to the p\_id-siblings.
Thus, when $q_2$ meets the $\varepsilon$-leaf of perso7
(viz. we know that the filter is false here)
it selects its second parameter, thus, we obtain
$\sem{q_2}(f_2',\sem{q_4}(f'))$ and proceed correctly.

%\smallskip

{\bf Size of a Transducer.}\quad
We define the size $|M|$ of the MFT $M$ as 
$|\Sigma|$ plus the sum of sizes of all left-hand sides
and right-hand sides of $M$'s rules. The size of
a forest is defined as the number of its nodes.

\section{From \minXQuery to MFT}

We now discuss the compilation of a \minXQuery program $P$ to the MFT $M_P$.
First, let us define a shorthand notation. For a forest $f$
that is restricted as the default rule of an MFT, but with the additional
restriction that $x_1,x_2$ do not appear, 
we denote by $q(\%,y_1,\dots,y_m)\to f$ the two rules
\[
\begin{array}{lcl}
q(\%t(x_1)x_2,y_1,\dots,y_m)&\to&f\\
q(\varepsilon,y_1,\dots,y_m)&\to&f\\
\end{array}
\]
We define $M_P=(Q,\Sigma,q_0,R)$ where $\Sigma$ consists of 
all element labels and string constants that appear in $P$.
For instance, for our example program $P_{\text{person}}$,
$\Sigma$ consists of person, p\_id, person0, and name.
The initial state $q_0$ of the MFT has the two rules induced by:
\[
q_0(\%)\to q_0'(x_0,q_{\text{copy}}(x_0))
\]
where $q_0',q_{\text{copy}}$ are states in $Q$ of ranks $2$ and
$1$, respectively.
The state $q_{\text{copy}}$ realizes the identity mapping on 
forests, via the rules
\[
\begin{array}{lcl}
q_{\text{copy}}(\%t(x_1)x_2)&\to&\%t(q_{\text{copy}}(x_1))q_{\text{copy}}(x_2)\\
q_{\text{copy}}(\varepsilon)&\to&\varepsilon.
\end{array}
\]
Our compilation functions are defined recursively on the
structure of $P$, and return sets of rules.
The compilation of any (sub)-expression of $P$ is done in 
the context of a mapping $\rho$ and a state $q\in Q$.
The mapping $\rho$ is of the form $\rho=\{(v_1,1),\dots,(v_n,n)\}$, where 
$v_i$ are variable names appearing in the \minXQuery program.
The state $q\in Q$ is the current state for which rules are
defined by the compilation. 
We define $\rho_0=\{($\|$input|$,1)\}$ and issue the call
$\mathcal{T}(P,\rho_0,q_0')$ as initial call to 
the compilation function $\mathcal{T}$.

We define the compilation function $\mathcal{T}$ recursively.
Let $e,e',e_1,\dots,e_n$ be arbitrary \minXQuery expressions,
$\rho$ a mapping, $q\in Q$ of rank $m+1$, $m\geq 0$,
and $p$ an XPath expression (as defined by the nonterminal $\text{\it ordpath}$
in the EBNF of Figure~\ref{fig:minxquery-syntax}).

If $e=e_1\cdots e_n$ then let $q_1,\dots,q_n$ be new states in $Q$ of
rank $m+1$ and define 
$\mathcal{T}(e,\rho,q)=\{r\}\cup\mathcal{T}(e_1,\rho,q_1)\cup\dots
\cup\mathcal{T}(e_n,\rho,q_n)$ where $r$ is the rule
$q(\%,y_1,\dots,y_m)\to q_1(x_0,y_1,\dots,y_m)\cdots q_n(x_0,y_1,\dots,y_m)$.

If $e=\|<|\sigma\|>| e'\|</|\sigma\|>|$ with $\sigma\in\Sigma$ 
then let $q'$ be a new state in $Q$ of
rank $m+1$ and define $\mathcal{T}(e,\rho,q)=\{r\}\cup\mathcal{T}(e',\rho,q')$ 
where $r$ is the rule
$q(\%,y_1,\dots,y_m)\to\sigma(q'(x_0,y_1,\dots,y_m))$.

If $e=\sigma$ (i.e., $\sigma$ is a string constant) then define
$\mathcal{T}(e,\rho,q)=\{q(\%,y_1,\dots,y_m)\to\sigma(\varepsilon)\}$.

If $e=\|$|v$ where $\|$|v$ is a variable name, then define
$\mathcal{T}(e,\rho,q)=\{q(\%,y_1,\dots,y_m)\to y_{\rho(\|$|v)}\}$.

If $e=\|for|~\|$|v~\|in|~p ~ e'$ then
let $q'$ be a new state in $Q$ of
rank $m+2$ and define 
$\mathcal{T}(e,\rho,q)=\mathcal{T}(e',\rho',q')
\cup\mathcal{F}(p,q,q')$
where $\rho'=\rho\cup\{(\|$|v,m+1)\}$
and $\mathcal{F}(p,q,q')$ is defined below.

If $e=\|let|~\|$|v\|:=|e_v ~ e'$ then
let $q_v,q'$ be new states in $Q$ of rank $m+1$ and $m+2$, respectively.
Define 
$\mathcal{T}(e,\rho,q)=\{r\}\cup\mathcal{T}(e_v,\rho,q_v)\cup\mathcal{T}(e',\rho',q')$
where $\rho'=\rho\cup\{(\|$|v,m+1)\}$ and
$r$ is the rule $q(\%,y_1,\dots,y_m)\to 
q'(x_0,y_1,\dots,y_m,q_v(x_0,y_1,\dots,y_m))$.

If $e=p$ with an XPath expression $p$,
then 
let $q'$ be a new state in $Q$ of
rank $m+2$ and define 
$\mathcal{T}(e,\rho,q)=\{r\}\cup\mathcal{F}(p,q,q')$
where $r$ is the rule
% $q'(\%,y_1,\dots,y_{m+1})\to q_{\m{copy}}(x_0)$.
$q'(\%,y_1,\dots,y_{m+1})\to y_{m+1}$.

The rules in $\mathcal{F}(p,q,q')$ are defined so that
\begin{equation}
\begin{array}{l}
\sem{q}(t\,\m{s}, u_1,\dots,u_m) =\\
\quad
\sem{q'}(t_1\,\m{s}_1, u_1,\dots,u_m,t_1)
\dots
\sem{q'}(t_n\,\m{s}_n, u_1,\dots,u_m,t_n)
\end{array}
\label{eqn:pathffor}
\end{equation}
where $t_1,\dots,t_n$ are all subtrees of $t$,
in pre-order,
that satisfy the XPath $p$ relative to the root of $t$,
and $\m{s}_1,\dots,\m{s}_n$ are the sequences of their following siblings.

Let us give a definition of $\mathcal{F}(p,q,q')$
for an XPath $p$,
two states $q$ and $q'$ with rank $m$ and $m+1$, respectively,
so that equation~(\ref{eqn:pathffor}) holds.
We first show the case where $p$ contains no predicate.
We obtain a total deterministic finite automaton~(DFA)
from the XPath $p$ in the usual way.
We only discuss child and descendant axes.
This translation is described 
by Green et al~\cite{DBLP:journals/tods/GreenGMOS04}. 
The cases for sequences of following-sibling axes is similar.
Without loss of generality,
the initial state of the DFA has no incoming transition.
The set $\mathcal{F}(p,q,q')$ consists of rules
each of which corresponds to a transition
$q_1\stackrel{a}{\longrightarrow}q_2$ of the DFA.
When $q_1$ is not initial and $q_2$ is not final,
$\mathcal{F}(p,q,q')$ contains a rule
$q_1(a(x_1)x_2,y_1,\dots,y_m)\to 
q_2(x_1,y_1,\dots,y_m) q_1(x_2,y_1,\dots,y_m)$.
When $q_1$ is initial and $q_2$ is not final
the set has a rule
$q(a(x_1)x_2,y_1,\dots,y_m)\to q_2(x_1,y_1,\dots,y_m)$.
When $q_1$ is not initial and $q_2$ is final,
the set has a rule
$q_1(a(x_1)x_2,y_1,\dots,y_m)\to 
q'(x_0,y_1,\dots,y_m,a(q_{\m{copy}}(x_1)))$.
When $q_1$ is initial and $q_2$ is final,
$q(a(x_1)x_2,y_1,\dots,y_m)\to 
q'(x_1,y_1,\dots,y_m,a(q_{\m{copy}}(x_1)))$.

Next we show the case where the XPath $p$ contains predicates.
We fist construct a set of rules in a way similar to the above
ignoring all predicates.
If a step in the XPath $p$ has a predicate $p'$,
we modify rules corresponding
to the transition for the step in the DFA.
For example, when $p=\|$|v\|//a[|p'\|]/b/c|$,
we modify the rule for the transition $q_1\stackrel{a}{\longrightarrow}q_2$ of the DFA
using another DFA obtained from the predicate XPath $p'$.
Before the modification,
we introduce a state $q_{p'}$ in the translated MFT
so that $\sem{q_{p'}}(t~ts,u_1,u_2) = u_1$
if the predicate XPath $p'$ is true for $t$ relative to the root of $t$,
and $\sem{q_{p'}}(t~ts,u_1,u_2) = u_2$ otherwise.
The set of rules for $q_{p'}$ is obtained in a way similar
to regular lookahead removal in macro tree transducers~\cite{DBLP:journals/jcss/EngelfrietV85}.
We use this state $q_{p'}$ for the modification of rules.
For example,
suppose that we obtain the following rules for $q_1$ by ignoring predicates:
\begin{align*}
&q_1(a(x_1)x_2,y_1,\dots,y_m)
\\&\quad\to q_2(x_1,y_1,\dots,y_m)~ q_1(x_2,y_1,\dots,y_m) 
\\&q_1(\%t(x_1)x_2,y_1,\dots,y_m) 
\\&\quad\to q_3(x_1,y_1,\dots,y_m)~ q_1(x_2,y_1,\dots,y_m)
\end{align*}
Then we modify the first rule as follows.
\begin{align*}
&q_1(a(x_1)x_2,y_1,\dots,y_m) 
\\&\quad\to
q_{p'}(x_1,\!\!\begin{array}[t]{l}
  q_2(x_1,y_1,\dots,y_m),\\
  q_3(x_1,y_1,\dots,y_m))~ q_1(x_2,y_1,\dots,y_m)
\end{array} 
\end{align*}

Let us summarize our translation using the example program $P_{\text{person}}$.
First, an MFT rule $q_0(\%)\to q_1(x_0, q_{\text{copy}}(x_0))$ is generated
for the initial state $q_0$, where $q_1$ corresponds to an expression 
$e=$\| <out>{|$e_{\text{for}}$\|}</out>|.
The accumulating parameter of $q_1$ is introduced for the \|$input| variable,
which is not used as an output hence it will be eliminated
in the further optimization.
For the $q_1$ state,
an MFT rule $q_1(\%,y_1)\to\text{out}(q_2(x_0,y_1))$ is generated 
by $\mathcal{T}(e,\rho,q_1)$ with $\rho=\{($\|$input|$,1)\}$.
For more MFT rules,
we compute $\mathcal{T}(e_{\text{for}},\rho,q_2)=
\mathcal{T}(e_{\text{let}},\rho',q_3)\cup\mathcal{F}(p,q_2,q_3)$
where $e_{\text{for}}=$\|for $b in|~$p$~\|return|~$e_{\text{let}}$
and $\rho'=\rho\cup\{($\|$b|$,2)\}$.
MFT rules for the $q_3$ state are obtained by the further computation
so that the output of $q_3$ is the results of the $e_{\text{let}}$ expression
for the current node.
The computation of $\mathcal{F}(p,q_2,q_3)$ generates
MFT rules for the $q_2$ state
which collect a sequence of results of $q_3$ at the path $p$.
After the whole translation, an MFT with 14 states is finally generated.
By the parameter reduction discussed in Section~\ref{sect:opt},
the MFT $M_{\text{person}}$ will be obtained.

The size $|P|$ of a \minXQuery program $P$ is defined as the number of nodes
in its parse tree (according to our EBNF in Figure~\ref{fig:minxquery-syntax}).

\begin{theorem}
Given a \minXQuery program $P$, the MFT $M_P$ 
is constructed in time $O(|P|)$.
For every XML forest $f$ it holds that
$\sem{M_P}(f)=\sem{P}(f)$.
\end{theorem}

\section{Optimizations}
\label{sect:opt}

\subsection{Parameter Reduction}
Our translation generates a transducer which
includes many redundant parameters in general.
They should be eliminated as much as possible
because the number of parameters has a serious effect on
efficiency of streaming of the obtained transducers.
In an extreme case, we can eliminate all parameters.
This will also help us to apply composition laws discussed in
Section~\ref{sect:comp}.
In this section, we suppose that 
the index position of arguments (or parameters) of states is starting with zero,
e.g., the parameter $y_2$ in $q(x,y_1,y_2)$ is called the second parameter.

In our implementation,
we eliminated two kinds of parameters:
{\em unused\/} parameters and {\em constant\/} parameters.
Additionally, we eliminated parameters by removing stay moves
and unreachable states.
Since the optimizations may interact,
we apply them repeatedly.

%\smallskip
{\bf Unused parameter reduction.}\quad
An unused parameter is one that does not appear in 
the output, for any given input.
For example,
if we have five rules for the states $q$ and $q'$
\[\begin{array}{lcl}
q(\sigma(x_1)x_2,y_1,y_2)
&\to& \delta(q'(x_2,y_1,y_2))
\\
q(\%t(x_1)x_2,y_1,y_2)
&\to& \%t(q'(x_2,\delta(y_2),\sigma(y_2)))
\\
q(\varepsilon,y_1,y_2)
&\to& \sigma(y_2)
\\
q'(\%t(x_1)x_2,y_1,y_2)
&\to&
q(x_1,\varepsilon,y_1)
\\
q'(\varepsilon,y_1,y_2)
&\to&
\varepsilon
\text,
\end{array}\]
then the parameters $y_1$ of $q$ and $y_2$ of $q'$ are unused
because they never contribute to outputs of the transducer.
The second parameter $y_2$ of $q$ is obviously used for output
because of the third rule.
From this fact, the parameter $y_1$ of $q'$ may also be used
because it will be passed to $q$ as the second argument in the fourth rule.
A set of unused parameters are obtained by finding all necessary parameters
in the following algorithm.
Let us call a {\em bare occurrence\/} of $y_i$ in $e$
when $y_i$ occurs in $e$ but not in an argument of a state call in $e$.
The algorithm collects all necessary parameters
as a set $S\subseteq U$ with $U=\{(q,i)\mid q\in Q,~ 1\leq i\leq\m{rank}(q)-1\}$
% of pairs of states and natural numbers
so that $(q,i)\in S$ implies
that the $i$-th parameter of state $q$ appears in outputs.
{\small
\[\begin{array}{l}
\quad S := \{ (q,i) \mid\!\!
             \begin{array}[t]{l}
             \text{$y_i$ is a bare occurrence in the right-hand}\\
             \text{side of $q$ rule.}\} \end{array}\\
\quad \mathbf{until}~ \text{$S$ is no longer updated} ~\mathbf{do}\\
\quad\quad  S := S \cup \{ (q,i) \mid\!\!
                      \begin{array}[t]{l}
                      (q',i')\in S,~\\
                      \text{$e$ is the $i'$-th argument of a $q'$ call}\\
                      \text{in the right-hand side of $q$ rule},\\
                      \text{$y_i$ is a bare occurrence in $e$}\} \end{array}\\
\quad \mathbf{end}
\end{array}\]
}
This procedure always terminates because of finiteness of $U$.
Obviously, $U\setminus S$ is a set of unused parameters.
For each $(q,i)\in U\setminus S$,
the parameter $y_i$ can be eliminated from the left-hand side of the $q$ rules.
We also remove the $i$-th argument of the $q$-call in the right-hand sides of all rules.

%An unused parameter is a parameter which does not occur
%in any right-hand sides of the rules for the same state.
%For example,
%if we have only three rules for the state $q$,
%\[\begin{array}{lcl}
%q(\sigma(x_1,x_2),y_1,y_2)
%&\to& q(x_1,\varepsilon,y_2)~ \delta(q'(x_2,y_2))
%\\
%q(\%t(x_1,x_2),y_1,y_2)
%&\to& \%t(q'(x_2,\delta(y_2)))
%\\
%q(\varepsilon,y_1,y_2)
%&\to& y_2\text,
%\end{array}\]
%then the parameter $y_1$ can be eliminated from the left-hand side of the $q$ rules.
%We also remove the corresponding parameter from the call of the state $q$
%in the right-hand sides of all rules.

%\smallskip
{\bf Constant parameter reduction.}\quad
A constant parameter is a parameter which is always instantiated by
the same constant forest.
This can be found by checking
whether the parameter of the state in the right-hand sides of all rules is
either the specific constant forest
or the parameter of the same state in the left-hand side.
For example, let $f$ be an XML forest and consider rules
\[\begin{array}{lcl}
q(\sigma(x_1)x_2,y_1,y_2)
&\to& q(x_1,\varepsilon,y_2)~ \delta(q'(x_2,y_2))
\\
q(\%t(x_1)x_2,y_1,y_2)
&\to& q(x_1,y_1,y_2)~ \%t(q'(x_2,\delta(y_2)))
\\
q(\varepsilon,y_1,y_2)&\to& y_1
\\
q'(\%t(x_1)x_2,y_1)&\to& \delta(q(x_1,\varepsilon,x_2))
\end{array}\]
and no other rule contains $q$ in its right-hand side.
The parameter $y_1$ can be eliminated from the $q$-rules.
We replace all occurrences of $y_1$ with the constant $\varepsilon$
in the right-hand side of the $q$ rules,
i.e., the third rule in the example above becomes
$q(\varepsilon,y_2)\to \varepsilon$.

%\smallskip
{\bf Stay move removal.}\quad
Removing stay moves also contributes to parameter reduction
because it may remove states with parameters by inlining.
For example,
if we have a rule $q(\%,y_1,y_2) \to q'(x_0) y_1$,
then all occurrences of $q(x_i,e_1,e_2)$ in the right-hand sides of the rules
can be replaced by $q'(x_i)~ e_1$.
Since the state $q$ is discarded,
the number of parameters is consequently reduced.
Note that our translation only introduces stay rule of the form
of the $q(\%,\dots)\to f$ which are particularly easy to inline.
A general procedure for stay-move removal of similar transducers
is given in Theorem~31 of~\cite{DBLP:journals/acta/EngelfrietM03}.

%\smallskip
{\bf Unreachable state removal.}\quad
Removing unreachable states can reduce the number of parameters
for the same reason as the stay move removal.
When we construct the state-call dependency graph according to all rules,
it is obvious that unreachable states from the initial state are unused.
%\medskip
%\medskip
In an extreme case, the translation of a given \minXQuery program
introduces only redundant parameters,
which can be reduced by the four procedures above.
It is possible to detect whether the case happens or not
from a \minXQuery program without translation.
Let us classify occurrences of variables in the \minXQuery program into three:
bound variables, path variables, and output variables.
A {\em bound variable\/} occurs at the left-hand side of a \|let| or \|for| clause;
a {\em path variable\/} occurs at the beginning of an XPath expression;
an {\em output variable\/} occurs at the other parts.
Our translation introduces parameters for two purposes:
XPath predicates and output variables.
Many parameters introduced for variable bindings can be removed
because most variables in \minXQuery programs occur as path variables.
%only for the XPath expression rather than for output.
These parameters are removed as unused parameters in the aforementioned way.
%The parameter introduced by a unused variable in the program
%can always be removed.
Additionally, 
if an output variable occurrence in the program is
only where it is introduced by the nearest enclosing \|for| clause,
the corresponding parameter can be removed by stay move removal.
In summary, we easily obtain the following lemma from these observation.

An MFT where each state is of rank $1$, i.e., in which no context
parameters $y_1,\dots$ are used is called \emph{top-down forest transducers},
abbreviated FT.

\begin{theorem}\label{th:noparam}
Let $P$ be a \minXQuery program.
When $P$ satisfies
(1) every XPath expression contains no predicates, and
(2) every output variable occurrence is not
inside of a \|for| clause
except that the corresponding bound variable occurrence is in the \|for| clause,
there effectively exists an FT equivalent to $\sem{P}$.
\end{theorem}
\begin{proof}
Suppose that a \minXQuery program $P$ satisfies the conditions above.
It suffices to show that
all accumulating parameters of the translated MFT can be removed.
From the first condition on XPath expressions,
all parameters of the translated MFT are introduced
in the following four cases of translation:
the initial state, 
\|for| clauses, \|let| clauses, and XPath queries.
For the initial state,
our translation introduces a parameter for the \|$input| variable.
From the second condition,
it does not occur inside any \|for| clause,
hence the parameter can be eliminated
in all states introduced for translating the inside expressions.
For the other states which have the parameter,
we can remove them by stay move removal.
As for the \|for| and \|let| clauses,
the present statement can be shown in a similar way.
The state introduced for an XPath query translation
can be immediately eliminated by stay move removal.
\end{proof}

\subsection{Composition}\label{sect:comp}
Composition of two XML transformations is to remove
intermediate XML tree constructions
like {\em deforestation\/}~\cite{DBLP:journals/tcs/Wadler90},
which has been heavily studied in the context of functional programming.
This can be a powerful optimization for XML processing.
Koch chose for GCX a compositional fragment of XQuery.
This means that two XQuery programs (where the second reads the
output of the first) can be composed into one
program. It can be shown that our fragment of XQuery can be composed as well.
What is known on the forest/tree transducer side with respect
to composition?
It is easy to see that
both MFT and MFT without parameters (FT) are not closed under composition.
But, two FTs can be composed into one MFT.
This can be obtained through known results; we give a direct construction
and determine its worst-case time complexity. 
We consider further composition results, when one of the involved
transducer is a \emph{tree transducer}, and state the complexity in terms
of bigO-notation. Our transducers are slightly different from those in 
the literature, plus, no complexity statements are known; therefore 
dwell on the theory and establish these results here.
We denote $f\comp g$ for a composition of two functions and
$F\comp G$ for a composition of two classes, that is,
$F\comp G=\{f\comp g \mid f\in F, g\in G\}$.

%\smallskip

{\bf Expressive Power.}\quad
An XML forest can naturally be seen as a
\emph{binary tree}, using the well-known first-child next-sibling
encoding (see, e.g.,~\cite{DBLP:journals/jcss/Schwentick07}).
In this encoding, the first child of an unranked node becomes the left child in the
binary tree, and the next sibling in the unranked tree becomes the right child
in the binary tree.
If an element node has no first child or no next sibling, then in the binary
tree it has the empty tree $\varepsilon$ as left (resp. right) child.

A \emph{binary XML tree} is a binary tree with internal nodes of rank $2$
labeled by elements in ${\mathbf U}^*$ and leaves labeled $\varepsilon$.
The set of all binary XML trees is denoted by $\mathcal B$.
For an XML forest $f\in{\mathcal F}$ we denote by $\text{fcns}(f)$ its first-child next-sibling
encoded binary XML tree in $\mathcal B$; i.e.,
$\text{fcns}(\varepsilon)=\varepsilon$ and for forests $f_1,f_2$
and $\sigma\in{\mathcal U}^+$,
\[
\text{fcns}(\sigma(f_1),f_2)=\sigma(\text{fcns}(f_1),\text{fcns}(f_2)).
\]
Given an MFT $M$, its \emph{binary tree translation} 
$\sem{M}_{\mathcal B}$ is the function over $\mathcal B$ defined as
\[
\sem{M}_{\mathcal B}=\{(\text{fcns}(f),\text{fcns}(g))\mid(f,g)\in\sem{M}\}.
\]
We denote by $\mft$ the class of all binary tree translations realized
by MFTs.

We can now compare the expressive power of MFTs to other well-known
classes of tree translations.
A \emph{macro tree transducer} (\emph{top-down tree transducer}),
for short MTT (TT),
is an MFT (FT) $M$ such that the right-hand side
of each rule is a {\bf tree} in which $(\Sigma\cup\{\%t\})$-labeled nodes
are binary.
In this case the output is always a binary tree, and therefore we define the
tree translation of $M$ as 
$\sem{M}_{\mathcal B}=\{(\text{fcns}(f),g)\mid(f,g)\in\sem{M}\}$.
The classes of translations are denoted $\mtt$ and $\tt$.
Macro and top-down tree transducers are conventionally defined for
ranked input and output alphabets (not necessarily binary), 
and do not have stay moves or default rules.
These inclusions hold:
%\[
$\tt\subsetneq\ft\subsetneq\mtt\subsetneq\mft$.
%\]

It was shown in~\cite{DBLP:journals/ipl/PerstS04}, for transducers
without stay and default rules, that every macro
forest transducer can be decomposed into a macro tree
transducer, followed by an ``evaluation mapping'' ${\sf eval}$.
The macro tree transducer is obtained from the macro forest transducer
by replacing each occurrence of concatenation
in the right-hand sides of the rules by a special binary symbol $@$.
For instance, the MFT right-hand side 
$q(x_1)y_1b(\varepsilon,\varepsilon)$ is replaced by the tree
$@(q(x_1),@(y_1,b(\varepsilon,\varepsilon)))$.
The evaluation mapping interprets $@$-symbols by concatenation, i.e.,
$\eval(@(t_1,t_2))=\eval(t_1)\eval(t_2)$, and for all other 
labels realizes the identity.
It should be clear that this result also holds in the presence of
stay and default rules. Thus, we have
%\[
$\mft\subseteq\mtt\comp\EVAL$.
%\]
It is not difficult to show that also the converse inclusion holds:
given an MTT $M$ and a evaluation mapping $\eval_\Sigma$, we can
construct an MFT $N$ such that $\sem{N}=\sem{M}\comp\eval_\Sigma$
(we simply remove all $@$-symbols from the right-hand sides of
$M$'s rules by interpreting them according to $\eval_\Sigma$).
The constructions do not affect the presence of parameters, and
thus the inclusions also hold for forest transducers 
(without context-parameters).
It is shown in~\cite{DBLP:journals/ipl/PerstS04} that 
$\eval_\Sigma$ can be realized by a macro tree transducer.
\begin{lemma}\label{lm:expr}
The following relations hold (and one
representation can be obtained from the other in linear time):
\begin{tabular}[t]{llll}
(1)&$\mft$&=&$\mtt\comp\EVAL$\\
(2)&$\ft$&=&$\tt\comp\EVAL$\\
(3)&$\EVAL$&$\subsetneq$&$\mtt$.
\end{tabular}
\end{lemma}

We want to derive new composition results for MFTs,
using existing results about tree transducers.
We are interested in complexity, and therefore must look
carefully how stay rules and default rules behave 
under composition.

As it turns out, stay rules are quite useful for
transducer composition: they allow to ``compress'' new
right-hand sides (using the compression power of 
transducer rules).
Without them, composing two top-down tree transducers takes exponential time,
with them: quadratic time!
This is easy to see:
consider a transducer $M_1$ that translates every $a$-node into
$4$ $b$-nodes:
\[
\begin{array}{lcl}
q_0(a(x_1))&\to&b(b(b(b(q_0(x_1)))))\\
q_0(\varepsilon)&\to&\varepsilon
\end{array}
\]
The next transducer $M_2$ spawns two new copies
for each $b$ node, via a rule of the form
\[
\begin{array}{lcl}
p_0(b(x_1))&\to&c(p_0(x_1),p_0(x_1))\\
p_0(\varepsilon)&\to&\varepsilon
\end{array}
\]
If we follow the natural product construction of translating via $M_2$ 
the right-hand sides of $M_1$'s rules,
then we obtain this deterministic top-down tree transducer (DT for short) rules
\[
\begin{array}{lcl}
\angl{q_0,p_0}(a(x_1))&\to&c(c(c(c(\angl{q_0,p_0}(x_1),\dots))))\\
\angl{q_0,p_0}(\varepsilon)&\to&\varepsilon
\end{array}
\]
which contains a complete binary tree of height $5$ in its right-hand side
and thus is of exponential size. 
It is not difficult to see
that there is no smaller equivalent DT.
In the presence of \emph{stay rules} we can avoid such
blow-ups. 
A stay transducer for the example does not have a right-hand side of exponential size
for the $(\angl{q_0,p_0},a)$-rule, but instead breaks up that tree into many
separate rules of the node-by-node $M_2$-translation of $M_1$'s $(q_0,a)$-rule:
\[
\begin{array}{lcl}
\angl{q_0,p_0}(a(x_1))&\to& c(\angl{q_0,p_0,1}(x_1),\angl{q_0,p_0,1}(x_1))\\
\angl{q_0,p_0,1}(a(x_1))&\to& c(\angl{q_0,p_0,2}(x_0),\angl{q_0,p_0,2}(x_0))\\
\angl{q_0,p_0,2}(a(x_1))&\to& c(\angl{q_0,p_0,3}(x_0),\angl{q_0,p_0,3}(x_0))\\
\angl{q_0,p_0,3}(a(x_1))&\to& c(\angl{q_0,p_0}(x_0),\angl{q_0,p_0}(x_0))\\
\end{array}
\]
Using stay rules we can construct in quadratic time
a DT realizing the composition of two given DTs.
In fact, this also works for two TTs, i.e.,
if the given transducers have stay rules and default rules.
Recall that the size $|M|$ of transducer $M$ is defined as
$|\Sigma|$ plus the sum of sizes of $M$'s rules.

\begin{lemma}\label{lm:compTT}
Let $M_1,M_2$ be TTs over $\Sigma$.
A TT $M$ can be constructed in
time $O(|\Sigma||M_1||M_2|)$ such that $\sem{M} = \sem{M_1} \comp \sem{M_2}$.
\end{lemma}
%\ifFULL
\begin{proof}
Let $M_i=(Q_i,\Sigma,q_i,R_i)$.
We first add some rules to $M_1$:
For every $a\in\Sigma$ for which there is
a $(p,a)$-rule in $R_2$ but no $(q,a)$-rule in $R_1$
we add the rule $r_a$ to $R_1$; the rule $r_a$ is 
obtained from $M_1$'s binary default rule for state $q$
by replacing every occurrence of $\%t$ (in left and right-hand side)
by $a$.
We define $M=(Q,\Sigma,\angl{q_1,q_2},R)$.
For all states $q\in Q_1$ and $p\in Q_2$ let
$\angl{q,p}$ be a state in $Q$.
For every rule $r\in R_1$, 
node $u$ of the right-hand side of $r$, and 
state $p\in Q_2$, let $\angl{r,u,p}$ be 
a state in $Q$.
Let $r$ be the rule 
$q(b(x_1,\dots,x_k))\to t$ with $k\in\{0,1,2\}$ and
$b\in\Sigma\cup\{\%t\}\cup\{\varepsilon\}$
and let $p$ be a state in $Q_2$.
We let the rule
%\[
$\angl{q,p}(b(x_1,\dots,x_k))\to\angl{r,\lambda,p}(x_0)$
%\]
be in $R$. 
Recall that $\lambda$ denotes the root node of a tree.
For every node $u$ of $t$ we let the rule 
%\[
$\angl{r,u,p}(b(x_1,\dots,x_k))\to t'$
%\]
be in $R$.
If $u$ is labeled by $q'(x_i)$ for $q'\in Q$ and $0\leq i\leq k$,
then define $t'=\angl{q,p}(x_i)$.
Otherwise,
$t'$ is obtained from the right-hand side of the unique
$p$-rule that is applicable to node $u$ of $t$.
We finally replace every $p'(x_i)$ by $\angl{r,u.i,p'}(x_0)$,
where $u.0$ denotes $u$.
The correctness of the construction follows from the fact that
for every input forest $t$:
$\sem{\angl{r,u,p}}(t) = \sem{p}(s)$ where $s$ is the subtree at
$u$ of $\sem{q}(t)$ and $r$ is the unique $q$-rule that is applicable
to the root of $t$. The statement can be proved by induction on 
the structure of $t$. 
Let $\text{maxrhs}(M_1)$ be the size of a largest right-hand side of $M_1$'s rules.
In the first step we add at most $|\Sigma|$-many rules to $R_1$.
We thus obtain a transducer of size $O(|\Sigma||M_1|)$.
For each rule $r$ of $M_2$ we construct in $M$ at most $O(|\Sigma||M_1|)$-many versions
of that rule (of same size as $r$). 
Thus $M$ is constructed in time $O(|\Sigma||M_1|||M_2|)$.
\end{proof}
%\fi % ifFULL

Note that the effective composition closure of total deterministic
top-down tree transducers (i.e., TT's without stay moves and default rules)
was proved in Theorem~2 of~\cite{DBLP:journals/mst/Rounds70};
it is also shown there that non-total such transducers are \emph{not} closed
under composition. Baker shows how to restrict nondeterministic
top-down tree transducers so that they can be composed into one
transducer~\cite{DBLP:journals/iandc/Baker79b}.
We are not aware of statements in the literature about 
the time complexity of tree transducer composition.
Before we give results about composition of forest transducers,
we lift two existing results about macro tree transducers 
to the presence of stay moves and default rules.

\begin{lemma}\label{lm:mttcomp}
Let $M_1$ be an MTT and $M_2$ a TT.
Then MTTs $M,M'$ can be constructed in
time $O(|\Sigma||M_1||M_2|)$ such that 
$\sem{M} = \sem{M_1} \comp \sem{M_2}$ and
$\sem{M'} = \sem{M_2} \comp \sem{M_1}$.
\end{lemma}
%\ifFULL
\begin{proof}
The construction of $M'$ is similar as in the proof of 
Lemma~\ref{lm:compTT}, so we omit the details.
\iffalse
We first discuss the construction of $M'$. 
For every state $q$ of $M_2$ and
$p$ of $M_1$ let $\angl{q,p}$ be a state of $M'$, and,
for every rule $r$ of $M_2$ and node $u$ in the right-hand side
of $r$ let $\angl{r,u,p}$ be a state of $M'$. 
The rank of such a state is defined to be the same as the
rank of $q$. The rules
are constructed in the same way as in the proof of 
Lemma~\ref{lm:compTT}, so we omit the details.
\fi
The construction of $M$ is more complicated.
Let $p_1,\dots,p_n$ be an ordering of the states of $M_2$.
Let $q$ be a state of $M_1$ of rank $m+1$, $m\geq 0$,
and let state $p_i$ be a state of $M_2$.
Then define $\angl{q,p_i}$ to be a state of $M$ of rank $m+1$. 
For every $q$-rule $r$ of $M_1$ and node $u$ in the right-hand side of $r$ 
define
$\angl{r,u,p_i}$ to be a state of $M$ of rank $1+mn$ with $n=|Q_2|$.
The idea is as before, state $\angl{r,u,p_i}$ is obtained by translating
the node $u$ of the right-hand side $t$ of $r$ in state $p_i$ of $M_2$.
The difference now is how to translate parameters $y_j$:
we must output the $p_i$-translation of the \emph{current} parameter tree 
in $y_j$. For this, we provide state $\angl{q,p}$ with $n$-many copies
of each parameter $y_j$ (one for each state $p_i$).
Details are omitted due to lack of space; they can be found
in the full version of the present paper~\cite{full}.
\iffalse
If $t$ at $u$ is labeled $y_j$, then we define 
\[
\angl{r,u,p_i}(b(x_1,\dots,x_k),y_1,\dots,y_{mn})\to
y_{(j-1)n+i}
\]
to be a rule of $M$.
If $t$ at $u$ is labeled $q'(x_e,t_1,\dots,t_{m'})$ then 
denote by $\_$ the sequence $y_1,\dots,y_{mn}$, let
$m''=m'+1$ and let
\begin{multline*}
\angl{r,u,p_i}(b(x_1,\dots,\_)\to\angl{q',p_i}(x_e,\\
\angl{r,u.2,p_1}(\_),\dots,\angl{r,u.2,p_n}(\_),\dots\\
\angl{r,u.m'',p_1}(\_),\dots,\angl{r,u.m'',p_n}(\_))
\end{multline*}
be a rule of $M$.
Otherwise, we can obtain the right-hand side
in the same way as the construction in the proof of Lemma~\ref{lm:compTT}
except that $mn$ parameters are carried.
Details are omitted due to lack of space; they can be found
in the full version of the present paper~\cite{full}.
\fi
\end{proof}
%\fi % ifFULL

Analogous results (without complexity statements)
about transducers without stay moves and 
default rules are stated in
Corollary~4.10 and Theorem~4.12 of~\cite{DBLP:journals/jcss/EngelfrietV85},
respectively.
%\ifFULL\else
%Lemmas~\ref{lm:compTT} and \ref{lm:mttcomp} can be proved by
%a straightforward product construction of two transducers,
%provided that preprocessing of transducers is required for
%stay moves and default rules.
%The detailed proofs are omitted because of the lack of the space.
%They can be found in the full version of the present paper~\cite{full}.
%\fi

%\smallskip

{\bf Composition of Forest Transducers.}\quad
We now consider the composition of two forest transducers (FTs), i.e.,
MFTs without accumulating parameters.
It is easy to see that FTs are \emph{not} closed under
composition:
(1) the output forests of any FT (seen as binary trees via the
first-child/next-sibling encoding) has height 
at most exponential in the height of the input tree.
(2) the composition of the following FT with itself has
double exponential height increase. It 
translates a forest of $n$ many $a$-nodes into a forest of 
$2^n$ many $a$-nodes:
\[
\begin{array}{lcl}
q(a(x_1,x_2))&\to&q(x_2)q(x_2)\\
q(\varepsilon)&\to&a.
\end{array}
\]

We now show that two FTs can be composed into one MFT. 
In fact, we show a stronger result: the composition 
of an MTT and an FT can be realized by one MFT.
Any FT can be turned in linear time 
into an equivalent MTT by turning each right-hand side into 
it binary tree encoding.
\begin{theorem}
Let $M_1$ be an MTT and $M_2$ an FT.
An MFT $M$ can be constructed in
time $O(|\Sigma||M_1||M_2|)$ such that 
$\sem{M} = \sem{M_1} \comp \sem{M_2}$.
\end{theorem}
\begin{proof}
By Lemma~\ref{lm:expr}(2), $M_2$ can be decomposed into
a TT $M_2'$ and an eval mapping $\eval_\Sigma$.
This takes time $O(|M_2|)$.
According to Lemma~\ref{lm:mttcomp} we construct in
time $O(|\Sigma||M_1||M_2'|)$ an mtt $M'$ with 
$\sem{M'}=\sem{M_1}\comp\sem{M_2'}$.
Finally, we compose $M'$ and the $\eval_\Sigma$ in time $O(|M'|)$ 
into the MFT $M$.
\end{proof}

\begin{theorem}
Let $\Sigma$ be an alphabet,
$M_1$ a TT over $\Sigma$, 
and $M_2$ an FT over $\Sigma$.
An FT $M$ can be constructed in
time $O(|\Sigma||M_1||M_2|)$ such that 
$\sem{M} = \sem{M_1} \comp \sem{M_2}$.
\end{theorem}
\begin{proof}
We decompose $M_2$ into $\tt\comp\eval$ in time $O(|M_2|)$
according to Lemma~\ref{lm:expr}(2).
We compose $M_1$ with the obtain TT in time $O(|\Sigma||M_1||M_2|)$.
From the obtained mapping in $\tt\comp\eval$ we construct
an FT again according to Lemma~\ref{lm:expr}(2), in linear time.
\end{proof}

\begin{theorem}
Let $M_1$ be an FT and $M_2$ an TT.
An MTT $M$ can be constructed in
time $O(|\Sigma||M_1||M_2|)$ such that 
$\sem{M} = \sem{M_1} \comp \sem{M_2}$.
\end{theorem}
\begin{proof}
We decompose $M_1$ into $\tt\comp\eval$ in time $O(|M_1|)$
according to Lemma~\ref{lm:expr}(2).
The eval-mapping can be turned into an MTT in linear time.
We compose the first TT and this MTT into one MTT,
according to Lemma~\ref{lm:mttcomp}. In this case we do not
need to add extra rules to the first transducer, and therefore
do not pay the $\Sigma$-factor. 
The reason is that the MTT for eval has no $a$-rules for $a\in\Sigma$
whatsoever: it only consists of default and $\varepsilon$-rules.
Finally, we compose the obtained
MTT with $M_2$ to obtain the desired MTT in time $O(|\Sigma||M_1||M_2|)$.
\end{proof}

\input{experiments}

\section{Conclusions}

We present a translation of XQuery fragments to MFTs.
This fragment is larger than the one supported by  GCX, the
fastest XQuery streaming tool we know.
The main difference to GCX is that we support XPath predicates 
and let-statements. Moreover, 
MFTs are more general and allow to easily program recursive
function definitions. 
Many useful static analyses are known for tree transducers
and can be applied to our MFTs. We applied three of them:
stay-move removal, useless parameter removal, and constant parameter
removal. The optimized MFTs are often faster by one order 
of magnitude in comparison to the unoptimized ones.
We present several efficient composition constructions for subclasses
of MFTs. These are useful so that intermediate stream results can
be avoided. 
We believe that MFTs are a robust and appropriate intermediate compilation 
framework for streaming of XQuery.
In the future we plan to apply further static analyses known for MFTs.
%tree transducers in order to obtain more efficient MFTs. 
We want to experiment with partial transducers that are obtained by 
composing with a domain check that corresponds to a DTD or XML Schema.
We would like to minimize  MFTs by using (a relaxed version of) the
``earliest normal form'', similar to the one known for 
deterministic top-down tree transducers of Engelfriet, Maneth, and Seidl~\cite{DBLP:journals/jcss/EngelfrietMS09}.
An earliest transducer produces its output as early as possible
during translation.
We would like to compress the output trees produced by a transducer:
Macro forest transducers can have doubly exponential size increase. 
This means that the size of an output tree is $O(2^{2^{n}})$, where
$n$ is the size of the corresponding input tree.
Their outputs can, 
however, be represented using grammar-based compression in \emph{linear space} with
         respect to the input size~\cite{DBLP:conf/fossacs/ManethB04}. Thus,
the output stream is guaranteed to be of linear size. 
It is a challenging open question how to execute
an MFT over an input stream that is grammar-compressed.
It was recently shown by Maneth, Ordonez, and Seidl~\cite{mos} 
that top-down tree transducers can executed in constant memory
over DAG-compressed tree streams.
Their  DAG streams contain forward references to definitions that
appear later in the stream. Alternatively, a tree can be shredded
into small parallel streams, as considered by
Labath and Niehren~\cite{ln13}.
Last, we want to study parallel processing of streams by MFTs; 
parallel execution of MFTs has 
been considered by Morihata~\cite{DBLP:conf/aplas/Morihata11}.

\section*{Acknowledgment}
%The authors would like to 
We thank Michael Benedikt and anonymous reviewers
for their valuable comments.
This work was partially supported by JSPS KAKENHI Grant Number 25730002.
Sebastian Maneth was supported by the Engineering and Physical Sciences Research
Council project ``Enforcement of Constraints on XML streams'' (EPSRC EP/G004021/1).

%\ifFULL
\bibliographystyle{abbrv}
%\else
%\bibliographystyle{IEEEtran}
%\fi % FULL
\bibliography{lib} 
%\newpage\appendix
%\input{appendix}

\end{document}

\subsection{Subsection Heading Here}
Subsection text here.

\subsubsection{Subsubsection Heading Here}
Subsubsection text here.

\section{Conclusion}


\section*{Acknowledgment}

The authors would like to thank...



%\bibliographystyle{IEEEtran}

\end{document}

%% file: macros.tex
\usepackage{fancyvrb}
\def\|{\verb|}

%% file: book.pdf_t
\begin{picture}(0,0)%
\includegraphics{book.pdf}%
\end{picture}%
\setlength{\unitlength}{4144sp}%
\begingroup\makeatletter\ifx\SetFigFont\undefined%
\gdef\SetFigFont#1#2#3#4#5{%
  \reset@font\fontsize{#1}{#2pt}%
  \fontfamily{#3}\fontseries{#4}\fontshape{#5}%
  \selectfont}%
\fi\endgroup%
\begin{picture}(3655,934)(6371,-4535)
\put(9181,-4471){\makebox(0,0)[b]{\smash{{\SetFigFont{10}{12.0}{\familydefault}{\mddefault}{\updefault}{\color[rgb]{0,0,0}Art of Programming}%
}}}}
\put(7561,-3751){\makebox(0,0)[b]{\smash{{\SetFigFont{10}{12.0}{\familydefault}{\mddefault}{\updefault}{\color[rgb]{0,0,0}book}%
}}}}
\put(7921,-4111){\makebox(0,0)[b]{\smash{{\SetFigFont{10}{12.0}{\familydefault}{\mddefault}{\updefault}{\color[rgb]{0,0,0}author}%
}}}}
\put(8551,-4111){\makebox(0,0)[b]{\smash{{\SetFigFont{10}{12.0}{\familydefault}{\mddefault}{\updefault}{\color[rgb]{0,0,0}title}%
}}}}
\put(7201,-4111){\makebox(0,0)[b]{\smash{{\SetFigFont{10}{12.0}{\familydefault}{\mddefault}{\updefault}{\color[rgb]{0,0,0}price}%
}}}}
\put(6571,-4111){\makebox(0,0)[b]{\smash{{\SetFigFont{10}{12.0}{\familydefault}{\mddefault}{\updefault}{\color[rgb]{0,0,0}isbn}%
}}}}
\put(7201,-4471){\makebox(0,0)[b]{\smash{{\SetFigFont{10}{12.0}{\familydefault}{\mddefault}{\updefault}{\color[rgb]{0,0,0}\$99}%
}}}}
\put(7921,-4471){\makebox(0,0)[b]{\smash{{\SetFigFont{10}{12.0}{\familydefault}{\mddefault}{\updefault}{\color[rgb]{0,0,0}Knuth}%
}}}}
\put(6571,-4471){\makebox(0,0)[b]{\smash{{\SetFigFont{10}{12.0}{\familydefault}{\mddefault}{\updefault}{\color[rgb]{0,0,0}123}%
}}}}
\end{picture}%

%% file: figure-minxquery-syntax.tex
%!TEX root = stream.tex
\newcommand\curlyL{\ensuremath{\text{\texttt{\textbf{\char123}}}}}
\newcommand\curlyR{\ensuremath{\text{\texttt{\textbf{\char125}}}}}
\begin{bnf}
\abovedisplayskip=0pt
\belowdisplayskip=0pt
  \begin{align*}
    \nonterm{query} \=&\mathrel{}
        \nonterm{element} \| \nonterm{clause} \\
    \nonterm{element} \=&\mathrel{}
          \begin{aligned}[t]
            &\term{<}\nonterm{elementname} 
%              \cstar{\group{\nonterm{attributename}
%                \term{=}\term{"}\nonterm{string}\term{"}}}
              \term{>}\\
            & \qquad \cstar{\group{\nonterm{element}
                \| \nonterm{string} 
                \| \curlyL\nonterm{clause}\curlyR
%                \| \curlyL\term{\$}\nonterm{var}\curlyR
                }}\\
            &\term{</}\nonterm{elementname}\term{>}
          \end{aligned}\\
    \nonterm{clause} \=&\mathrel{} 
        \ForClause{\nonterm{var}}{\nonterm{ordpath}}{\nonterm{query}}\\
      \|&\mathrel{}
        \LetClause{\nonterm{var}}{\nonterm{query}}{\nonterm{query}}\\
      \|&\mathrel{}
        \nonterm{ordpath}\\
      \|&\mathrel{}
        \term{(}\nonterm{query}\cplus{\group{\term{,}\nonterm{query}}}\term{)}\\
    \nonterm{ordpath} \=&\mathrel{}
        \term{\$}\nonterm{var} \cstar{\group{\nonterm{pathstep}}}\\
%        \cstar{\group{\group{\term{/}\|\term{//}}\nonterm{axis\_test}
%          \cstar{\term{[}\nonterm{predicate}\term{]}}}}\\
    \nonterm{pathstep} \=&\mathrel{}
%        \group{\term{/}\|\term{//}}\nonterm{axis}\term{::}\nonterm{nodetest}
        \term{/}\nonterm{axis}\term{::}\nonterm{nodetest}
        \cstar{\group{\term{[}\nonterm{predicate}\term{]}}}\\
    \nonterm{axis} \=&\mathrel{}
        \term{child} \| \term{descendant} \| \term{following-sibling}\\
%&\| \term{attribute::}\nonterm{attrtest}\\
    \nonterm{nodetest} \=&\mathrel{}
        \nonterm{elementname} \| \term{*} \| \term{text()} \| \term{node()}\\
%    \nonterm{attrtest} \=&\mathrel{}
%        \nonterm{attributename} \| \term{*} \\
    \nonterm{predicate} \=&\mathrel{}
        \nonterm{predpath} \| \term{empty}\term{(}\nonterm{predpath}\term{)}\\
      \|&\mathrel{}
        \nonterm{predpath}\term{=}\term{"}\nonterm{string}\term{"}
        \| \nonterm{predpath}\term{!=}\term{"}\nonterm{string}\term{"}\\
    \nonterm{predpath} \=&\mathrel{}
        \term{.} \cstar{\group{\nonterm{pathstep}}}
%        \cstar{\group{ \group{\term{/}\|\term{//}}
%        \nonterm{axis\_test}\cstar{\term{[}\nonterm{predicate}\term{]}}}}
  \end{align*}
\end{bnf}

%% file: experiments.tex
%!TEX root = stream.tex
\section{Experiments}
\input{appendix}
We have implemented in OCaml both our translation
from \minXQuery programs into MFTs 
and MFT optimizations.
In this section,
we present experimental results of our implementation
by connecting with the stream processor generators for MFTs
by Nakano and Mu~\cite{DBLP:conf/aplas/NakanoM06}.
All experiments are conducted on
an Apple XServe with 2.93 GHz 8-core Intel Xeon and 48 GB main memory.
In the experiments, 
we compare our implementation with
GCX~\cite{DBLP:conf/vldb/KochSS07,DBLP:conf/icde/SchmidtSK07} and Saxon~\cite{saxon},
both of which are stream processors for XQuery.
GCX supports only a subset of XQuery like ours,
while Saxon covers full features of XQuery.
The experiments shot
that our MFT-based approach achieves performance on a par with GCX.
We only show the numbers of comparison between ours and GCX
because Saxon is much slower than the two.
The reason comes from the fact that the Saxon's streaming is optimized
for memory, not for speed.
It is implemented in Java which has much overhead in loading
the Java VM and warming up the hotspot compiler when run from the command line.
A comparison with Saxon is not fair anyhow, because it supports
the full standards, while GCX and our engine only implement a small
subset of XQuery.

\newcommand\putgraph[2]{%
\subfigure[#1]{%
\begin{minipage}{.45\textwidth}\center
\hspace*{5ex}\includegraphics[scale=0.8]{#2}
\end{minipage}
}}
\newcommand\putgraphX[2]{%
\subfigure[#1]{%
\begin{minipage}{.30\textwidth}\center
\hspace*{3ex}\includegraphics[scale=0.8]{#2}
\end{minipage}
}}
\begin{figure*}
~\\[-10ex]
\center
\putgraphX{XMark Q1\label{fig:xmarkQ1}}{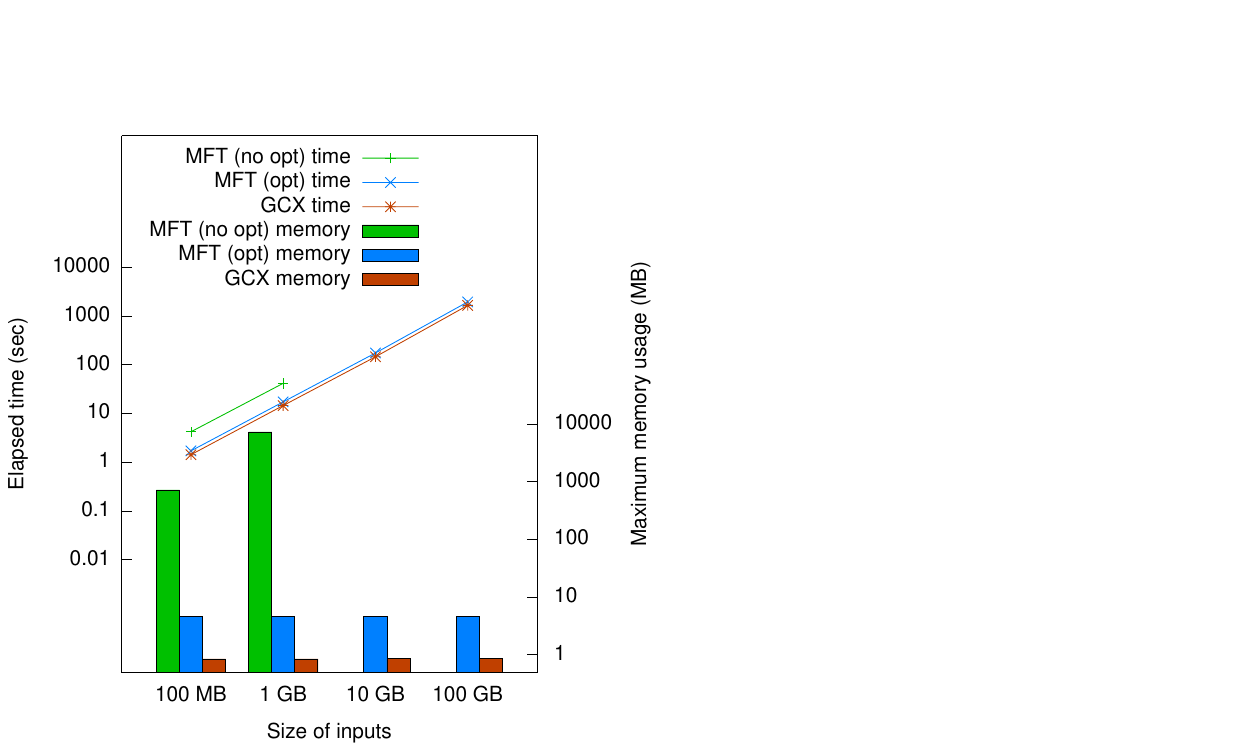}
\putgraphX{XMark Q2\label{fig:xmarkQ2}}{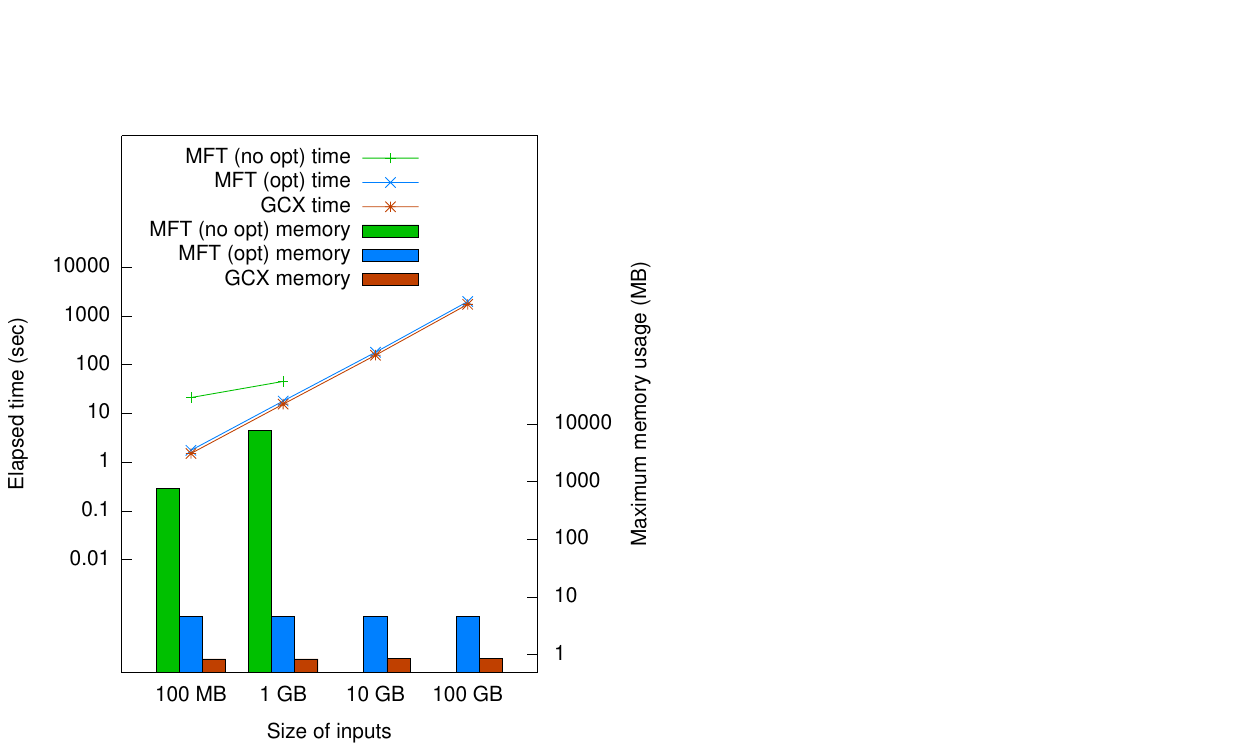}
\putgraphX{XMark Q4\label{fig:xmarkQ4}}{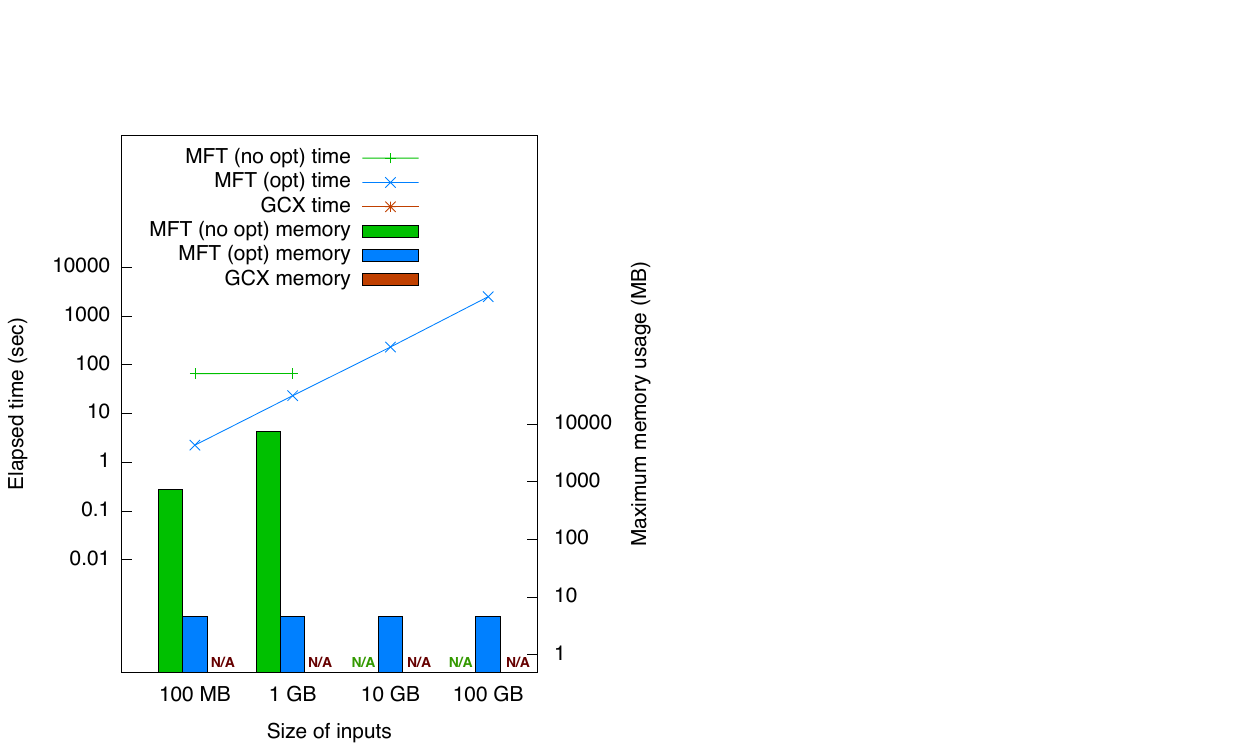}
\\[-3ex]
\putgraphX{XMark Q13\label{fig:xmarkQ13}}{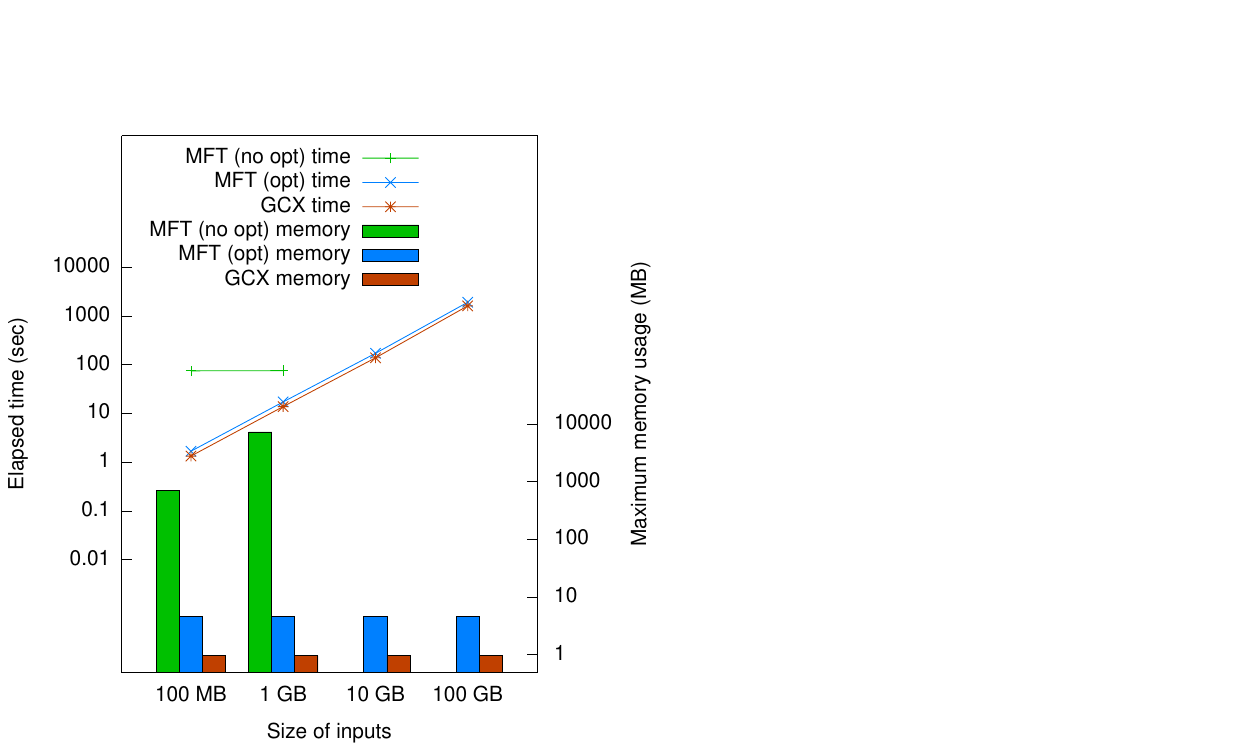}
%\ifFULL
%\putgraphX{XMark Q15\label{fig:xmarkQ15}}{xmark_q15.pdf}
%%\putgraph{XMark Q14 (with hand-written MFT)\label{fig:xmarkQ14}}{xmark_q14.pdf}
%\putgraphX{XMark Q16\label{fig:xmarkQ16}}{xmark_q16.pdf}
%\\[-3ex]
%\putgraphX{XMark Q17\label{fig:xmarkQ17}}{xmark_q17.pdf}
%\putgraphX{Identity query\label{fig:identity}}{identity-na.pdf}
%\else
\putgraphX{XMark Q16\label{fig:xmarkQ16}}{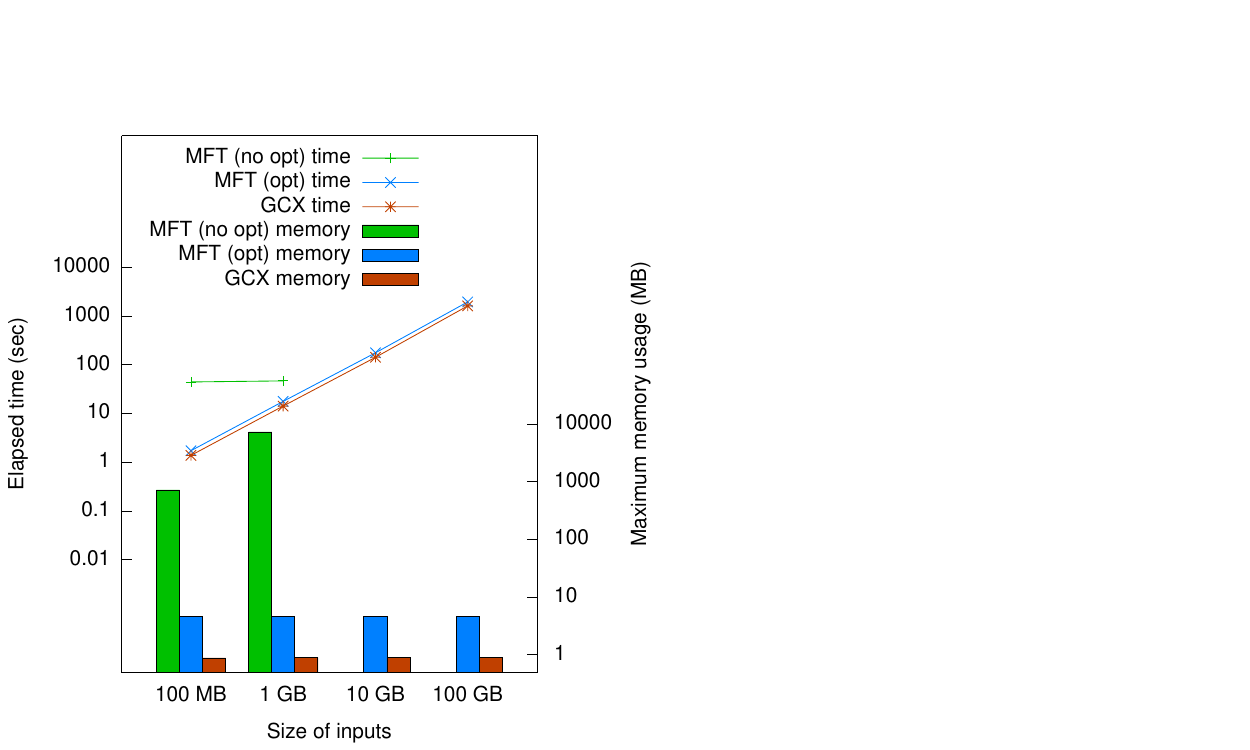}
\putgraphX{XMark Q17\label{fig:xmarkQ17}}{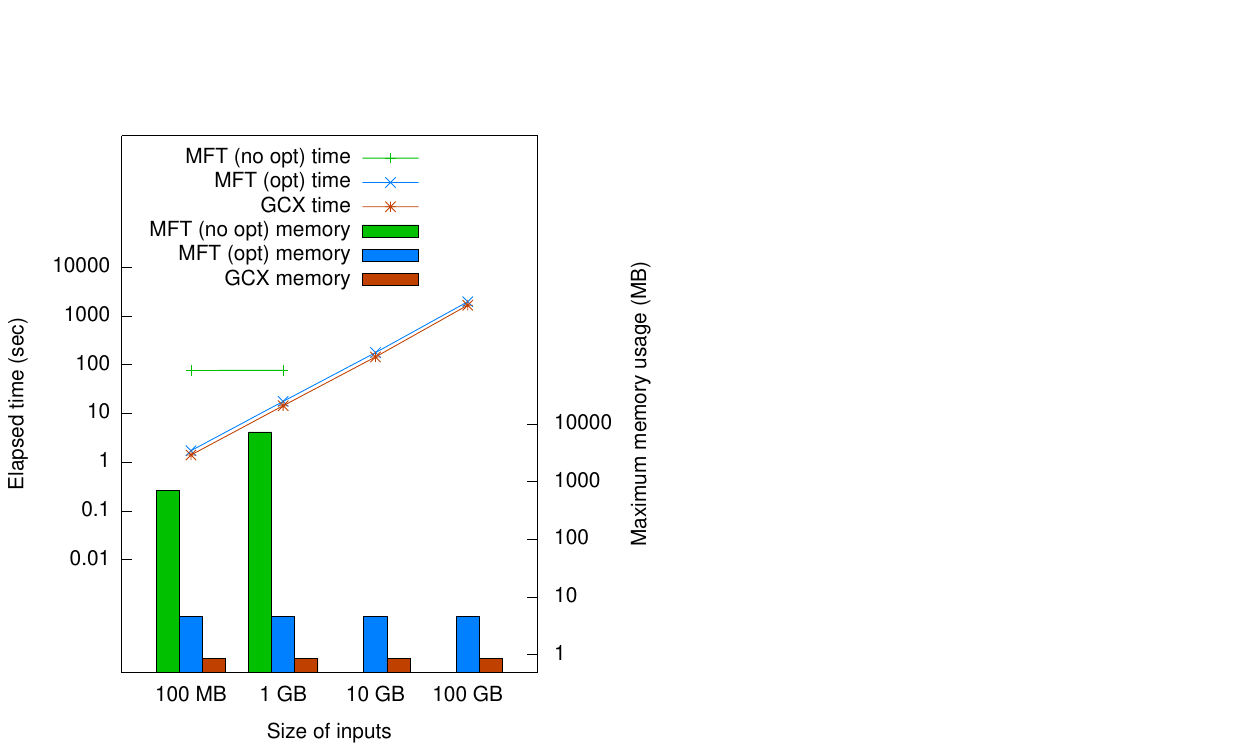}
\\[-3ex]
%\fi % FULL
\putgraphX{Double query\label{fig:double}}{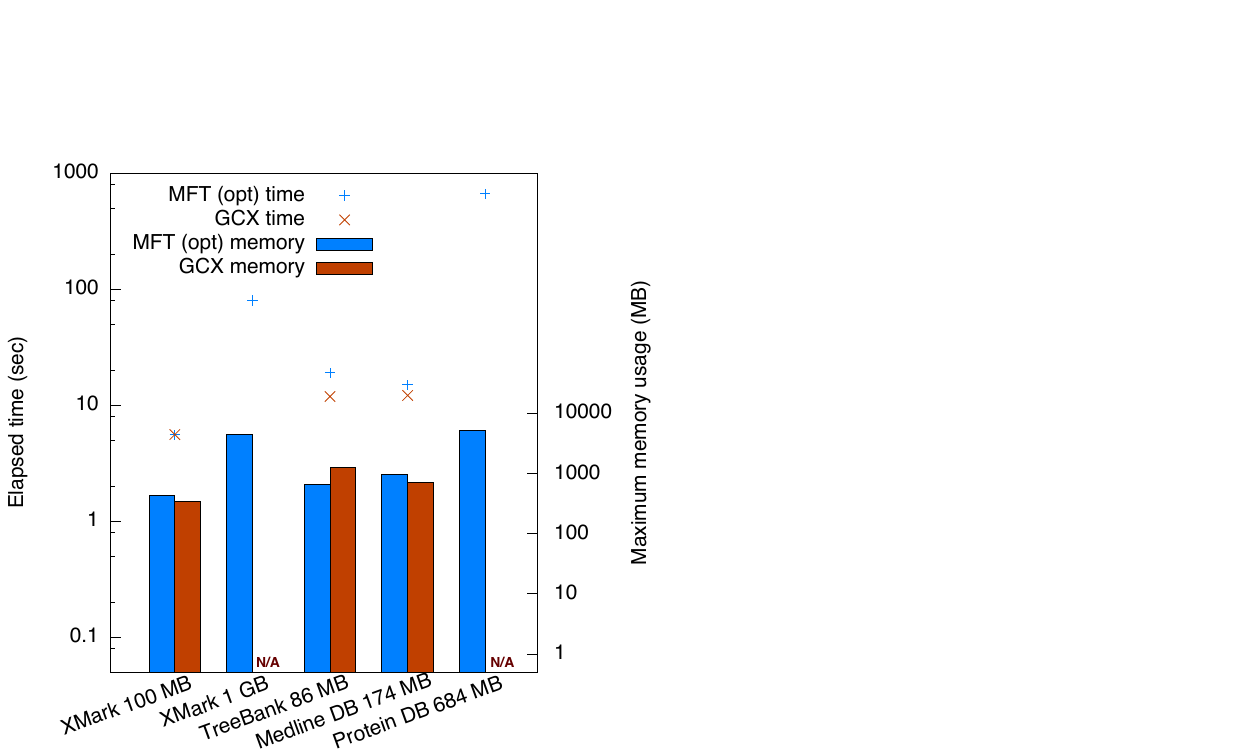}
%\hfil
\putgraphX{4star query\label{fig:4star}}{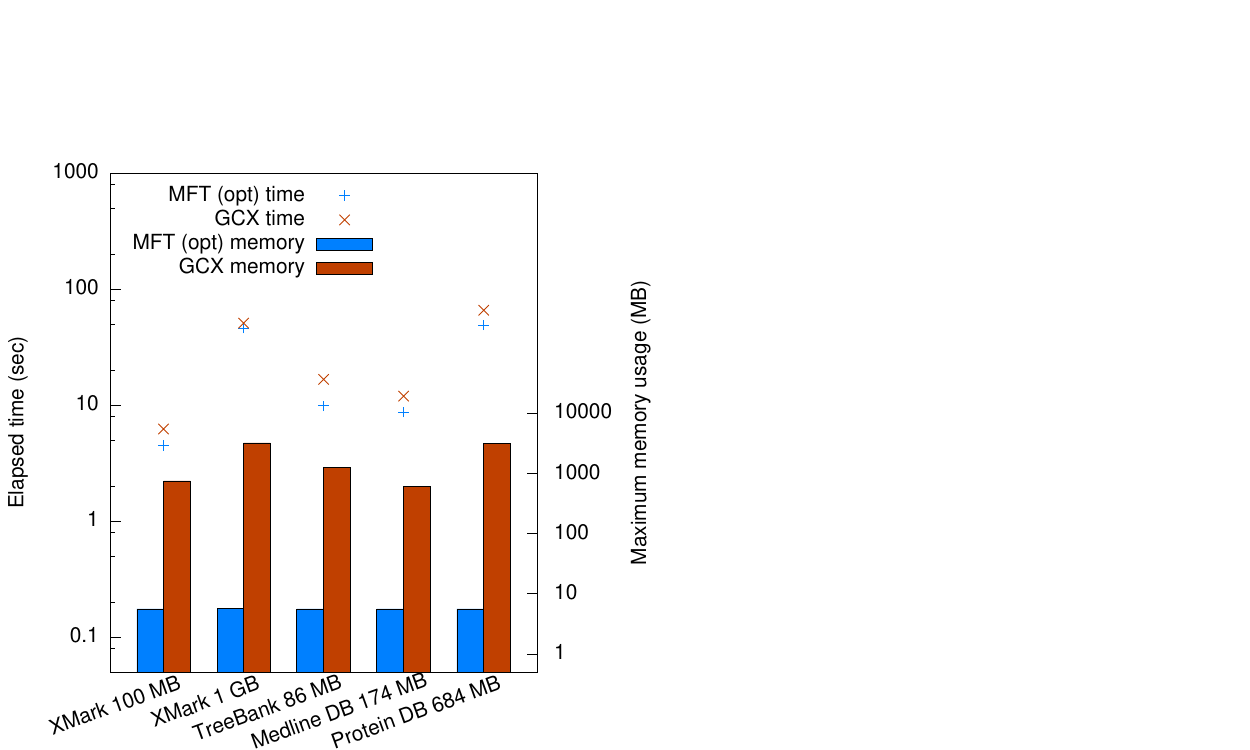}
%\hfil
\putgraphX{Deepdup query\label{fig:deepdup}}{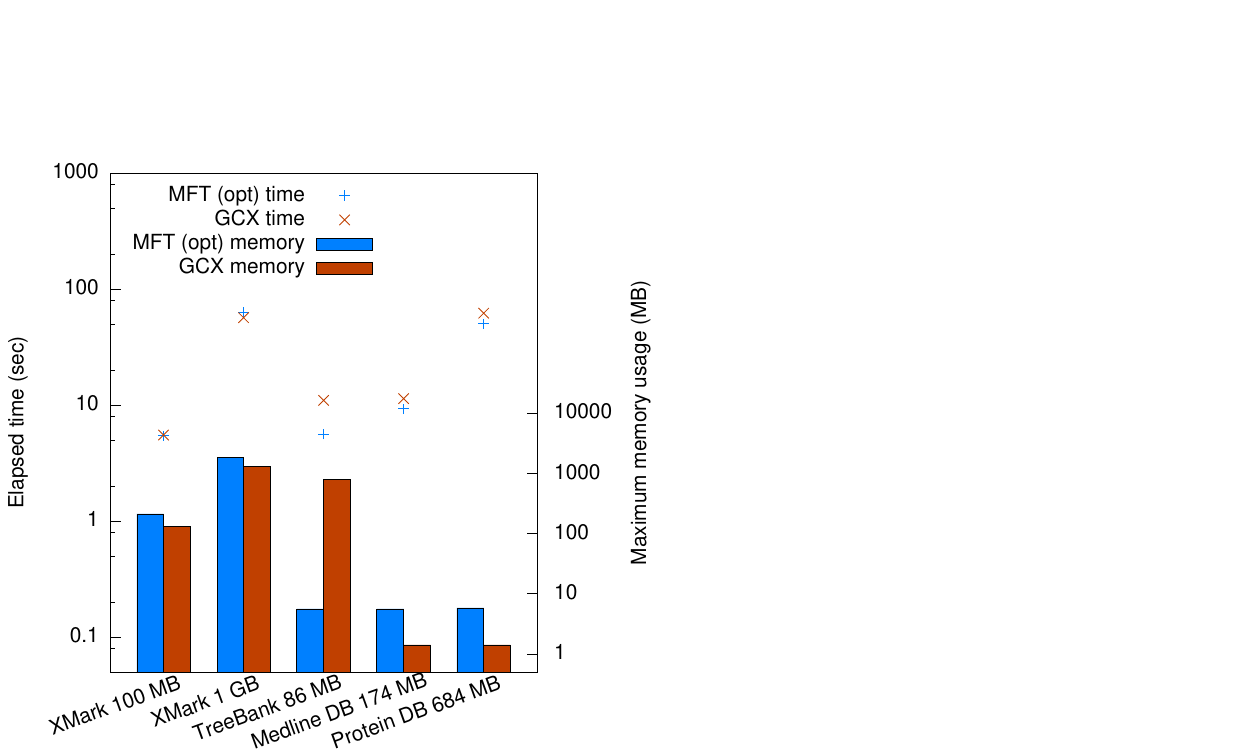}
\caption{Benchmark results}
\label{fig:experiments}
\end{figure*}

We run over XMark documents~\cite{DBLP:conf/vldb/SchmidtWKCMB02,xmark}
with sizes ranging from 100 MB to 100 GB for benchmarking.
The GCX distribution comes with example queries for XMark
which are adapted versions of the standard XMark XQuery benchmark queries.
These $20$ queries are of the following three different types:
\begin{enumerate}
\item[(1)] queries with small output (at most 1\% of the input size)
realizing simple XPath selection (downward axes only)
\item[(2)] queries producing huge output of quadratic size 
\item[(3)] queries that use data value joins
\end{enumerate}
Out of the $20$ queries, $16$ are of type (1). 
%\ifFULL
%From these we have picked the seven most diverse queries:
%Q1, Q2, Q4, Q13, Q15, Q16, and Q17.
%\else
From these we have picked the six most diverse queries:
Q1, Q2, Q4, Q13, Q16, and Q17.
%\fi %FULL
%We also benchmark the identity query and a simple XPath query.
Queries of type (2) might not be used in practice very often
(they are not quadratic in their original form of the XMark benchmark). 
%Here we run two sets of experiments:
%first, with normal XML output; our performance is comparable to GCX, and,
%second, producing compressed output. 
%Since the size of the compressed output is only linear in 
%the input size, we observe a drastic speed-up when we grammar-compress
%the outputs.
Since we do not support join yet, we do not test queries of type (3).
It is not difficult to add standard join procedures to our tool, and
we do not expect major time differences in comparison to GCX; this is
left for future work.

Figure~\ref{fig:experiments} shows 
results of the comparison between
streaming MFT and GCX
on total elapsed time and maximum memory consumption.
We ran some queries in XMark and simple examples of XQuery
with large-sized input XML data.
Missing data in the graphs indicates a failure of execution because of out-of-memory
except for Figure~\ref{fig:xmarkQ4} in which GCX fails because of the lack of expressiveness.
We also investigate performance of streaming MFTs translated from XQuery
before and after applying the optimizations discussed in Section~\ref{sect:opt},
which are referred by ``MFT (no-opt)'' and ``MFT (opt)'' in the graphs, respectively.
As can be seen, the optimized MFT and GCX consume constant-sized memory,
independent of the size of inputs.
On the other hand, the unoptimized MFT consumes much more space
and often fails to process large-sized inputs.
This is because the MFTs translated from XQuery contain many redundant parameters
for every variable
which may not be used for a part of the output.
Since an unoptimized MFT necessarily stores an entire input
for the \|$input| variable, it cannot benefit from streaming-style evaluation.
Therefore, the optimization phases are indispensable for our MFT-translation and
the used MFT stream processor.

The memory consumptions differ by a factor of about three
between the optimized MFT and GCX.
This is because of differences of their implementation languages
and the employed XML parsers.
Our streaming engine is written in OCaml, while GCX is implemented in C++.
The measured OCaml memory footprint (including the Expat XML parser) was about 4.5 MB,
which was the major factor of the difference.
Additionally, GCX could have advantages on memory usage
since it employs its own XML parser which does not handle attributes,
namespaces and character codes.
This difference also affects the comparison results on elapsed time.

The XQuery programs we used are listed in Figure~\ref{fig:queries}.
Although the \|//| step is not shown in the syntax of \minXQuery,
it is supported in our implementation in a usual way.
%The \|$input| variable is automatically inserted for absolute XPath expressions.
We modify every program for benchmark it on GCX
by replacing XPath predicates with \|where|-clauses.

Let us discuss the results for each query.
Figure~\ref{fig:xmarkQ1} shows the results of running XMark Q1
which requires a simple lookahead because of XPath predicates.
Our implementation is about 18\% slower than GCX.
Figure~\ref{fig:xmarkQ2} shows the results of running XMark Q2
which contains nested for-loops.
Since this query contains neither XPath predicates nor let-clauses,
all accumulating parameters can be removed,
i.e., the optimized MFT is in FT.
In this query, the elapsed times of our MFT-based approach are very close to GCX.
%Ours is about 14\% slower than GCX.
%GCX is 1.14 times faster than 
%
Figure~\ref{fig:xmarkQ4} shows the results of running XMark Q4
which requires selecting according to the sibling order.
It is an interesting example in the sense that
the query requires a nested XPath predicate like
\begin{verbatim}
  /site/open_auctions/open_auction
    [./bidder[./personref//text()="person111"]
      /following-sibling::bidder
               /personref//text()="person222"].
\end{verbatim}
GCX fails to run because the \|following-sibling| axis is not supported.
Figure~\ref{fig:xmarkQ13} shows the results of running XMark Q13
which requires reconstruction of XML data in the result.
Since this query satisfies the condition of Theorem~\ref{th:noparam},
%contains neither XPath predicates nor let-clauses,
all accumulating parameters can be removed,
i.e., the optimized MFT is an FT.
%In this query, the elapsed times of our MFT-based approach are very close to GCX.
Our implementation is about 20\% slower than GCX.
%
%Figure~\ref{fig:xmarkQ15} shows the results of running XMark Q15
%which involves a very long XPath expression.
%The optimized MFT for this query does not have accumulating parameters
%for the same reason as the previous one.
%%Since this query contains neither XPath predicates nor let-clauses,
%%all accumulating parameters can be removed,
%%i.e., the optimized MFT is in FT.
%%In this query, the elapsed times of our MFT-based approach are very close to GCX.
%%We can remove all accumulating parameters also in this query
%%because of the same reason as the previous one.
%%It shows results similar to others.
%
Our implementation is about 23\% slower than GCX.
Figure~\ref{fig:xmarkQ16} shows the results of running XMark Q16
which involves an XPath expression with a very long XPath predicate.
In this query, we need to very deep lookahead in order to select nodes
that may be harmful for stream processing.
Although our MFT translation introduces many states with rank 3,
the performance of our MFTs is still acceptable.
Figure~\ref{fig:xmarkQ17} shows the results of running XMark Q17
which involves a negative XPath predicate \|[empty(./homepage/text())]|.
%It shows results similar to others.
The experiments show results similar to others.
Our implementation is about 21\% slower than GCX.

\begin{table}\center
\caption{Input XML files for benchmark}
\label{tbl:inputs}
\begin{tabular}{|l|r|r|}
\hline
\multicolumn1{|c|}{\hspace*{30ex}}&
\multicolumn1{|c|}{\parbox[b]{12ex}{\centering size}}&
\multicolumn1{|c|}{\parbox[b]{12ex}{\centering depth}}
\\\hline
XMark & any & 13
\\\hline
TreeBank DB & 86 MB & 37 
\\\hline
Medline DB & 174 MB & 8
%Termination Problems DB & 120 MB & 77 
\\\hline
Protein Sequence DB & 684 MB & 8 
\\\hline
\end{tabular}\\[.5ex]
\hspace{20ex}{\scriptsize All attribute nodes are encoded as element nodes.\/}
\end{table}

%\ifFULL
%The last four queries test corner cases: 
%the identity mapping,
%the input doubling,
%the selection via the XPath query \|//*//*//*//*|,
%and deep duplication.
%\else
The last three queries test corner cases: 
the input doubling,
the selection via the XPath query \|//*//*//*//*|,
and deep duplication with a nested for-loop.
%\fi %FULL
Here our MFT-based implementation shows better results than GCX.
%\ifFULL
%Figure~\ref{fig:identity} shows the results of running the identity query
%which just outputs the same XML data as inputs.
%Our implementation is 187 \% faster than GCX at the maximum.
%Our implementation is 106 \% faster on average
%and 187 \% faster at the maximum than GCX.
%\fi %FULL
Figure~\ref{fig:double} shows the results of the input doubling query
that outputs the input XML twice,
which requires the entire input to be stored in memory for the second output.
This example shows the our implementation enables to run in streaming style
even such an extreme case.
GCX seems buggy on the input doubling query
(even the identity mapping like \|<out>{$input/*}</out>|),
which fails to process when the size of an input is larger than 200 MB.
Figure~\ref{fig:4star} shows the results of running the node selection 
with XPath \|//*//*//*//*|.
Our implementation is 26 \% faster on average
and 41 \% faster at the maximum than GCX.
Figure~\ref{fig:deepdup} shows the results of the query which require
a variable copying many times.
Our implementation is 15 \% faster on average
and 48 \% faster at the maximum than GCX.
We tried various kinds of XML data shown in Table~\ref{tbl:inputs}.
Treebank DB has very deep tree structures
even when the size is not so large.
%Protein Sequence Database which has the maximum depth 7
%and whose size is 684 MB;
%DBLP data which has the maximum depth 6 and whose size is 1.17 GB.
According to the results,
the depth affects the performance of stream processing even when the total size is small.
It is difficult to give a general statement
on when our streaming has an advantage over GCX.
We leave to future work more investigation.

%% file: appendix.tex
%!TEX root = stream.tex
%\newpage
%{\sc Appendix}
%\def\topfraction{.99}
%\section{Benchmark queries}
%\label{sec:allqueries}
%The XQuery programs we used are listed in Figure~\ref{fig:queries}.
%Although the \|//| step is not shown in the syntax of \minXQuery,
%it is supported in our implementation in a usual way.
%\vspace{-10ex}
%We modify every program for benchmark it on GCX
%by replacing XPath predicates with \|where|-clauses.

\begin{figure}[t]
%{\large\sc Appendix}\\[1ex]
%XQuery Benchmark Programs
%Although the \|//| step is not shown in the syntax of \minXQuery,
%it is supported in our implementation in a usual way.
{\scriptsize
\begin{Verbatim}[frame=single,label={Benchmark Queries}]
<query01>{
for $person in $input/site/people/person
               [./person_id/text()="person0"]
return $person/name/text()}</query01>

<query02>{
for $open_auction in /site/open_auctions/open_auction return
 <increase>{ for $increase in $open_auction/bidder/increase return
   <bid>{$increase/text()}</bid> }</increase>
}</query02>

<query04>{
for $b in $input/site/open_auctions/open_auction
          [./bidder[./personref/personref_person/text()="personXX"]
            /following-sibling::bidder/personref/personref_person
            /text()="personYY"]
return <history>{$b/reserve/text()}</history>}</query04>

<query13>{
for $item in $input/site/regions/australia/item
return <item><name>{$item/name/text()}</name>
             <description>{$item/description}</description></item>
}</query13>

<query16>{
for $closed_auction in $input/site/closed_auctions/closed_auction
                [./annotation/description/parlist/listitem/parlist
                  /listitem/text/emph/keyword/text()] return 
  <person><id>{$closed_auction/seller/seller_person}</id></person>
}</query16>

<query17>{
for $person in $input/site/people/person[empty(./homepage/text())] 
return <person><name>{$person/name/text()}</name></person>
}</query17>

<double><r1>{$input/*}</r1>{$input/*}</double>

<fourstar>{$input//*//*//*//*}</fourstar>

<deepdup>{ for $x in $input/* return 
 <r> { for $y in $x/* return <r1><r2>{$y}</r2>{$y}</r1> } </r>
}</deepdup>
\end{Verbatim}
}
%}
% removed queries
%<query15>{
%for $text in $input/site/closed_auctions/closed_auction/annotation
%             /description/parlist/listitem/parlist/listitem/text
%             /emph/keyword/text() return 
%<text>{$text}</text>}</query15>
%
%<identity>{$input/*}</identity>

\caption{XQuery Benchmark Programs}
\label{fig:queries}
\end{figure}